\documentclass[12pt]{article}
\usepackage[utf8]{inputenc}
\usepackage[T1]{fontenc}
\usepackage[sc]{mathpazo}
\usepackage[margin=1in]{geometry}
\usepackage{amsmath, amssymb, amsfonts, mathrsfs, bbm, xfrac,amsthm}

\usepackage[table]{xcolor}
\definecolor{lightshade}{gray}{0.95}
\definecolor{textgray}{gray}{0.4}
\definecolor{textdarkgray}{gray}{0.3}

\usepackage{graphicx, epstopdf, pdfpages}
\usepackage{float}
\usepackage[position=top]{subfig}
\usepackage{lscape}
\usepackage[onehalfspacing]{setspace}

\usepackage{tabularx, booktabs, threeparttable, tablefootnote, siunitx}
\usepackage{algorithm, algorithmic, chngcntr, comment, footmisc, natbib}

\usepackage[labelfont=bf]{caption}
\newtheorem{theorem}{Theorem}

\newtheorem{proposition}[theorem]{Proposition}

\newtheorem*{remark*}{Remark}

\usepackage[colorlinks=true, urlcolor=blue, linkcolor=blue, citecolor=blue]{hyperref}

\makeatletter
\def\proceduretype{}
\newcommand\PROCEDURE[3][default]{%
  \def\proceduretype{#1}%
  \ALC@it
  \textbf{#1}\ \textsc{#2} {#3}%
  \begin{ALC@prc}%
}
\newcommand\ENDPROCEDURE{%
  \end{ALC@prc}%
  \ifthenelse{\boolean{ALC@noend}}{}{%
    \ALC@it\algorithmicend\ \textbf{\proceduretype}%
  }%
}
\newenvironment{ALC@prc}{\begin{ALC@g}}{\end{ALC@g}}
\makeatother

\begin{document}

\title{\textbf{Fiscal Limits to Protectionism: \\The 2025 U.S.\ Tariff Laffer Curve}\thanks{We would like to thank Jevan Cherniwchan, Oliver Loertscher, Adam Lavecchia, and Tim Kehoe for insightful feedback on earlier results of this paper.  Pujolas thanks SSHRC for Insight Grants 435-2021-0006, 435-2024-0339, and for the Partnership to Study Productivity, Firms and Incomes. Pujolas: pujolasp@mcmaster.ca; Rossbach: Jack.Rossbach@georgetown.edu.} }
\author{  \textbf{Pau S. Pujolas}  \\ \small{McMaster University} \and \textbf{Jack Rossbach} \\ \small{Georgetown University in Qatar}}

\date{
\begin{tabular}{ccc}
&&\\
&February 2026&
\end{tabular}}
\maketitle
\thispagestyle{empty}

\begin{abstract}
\noindent
We quantify the U.S. tariff Laffer curve using a multi-sector Ricardian model calibrated to the 2025 U.S. trade war. We find revenue-maximizing tariffs of 20--30 percent and welfare-maximizing rates of 0--10 percent. We introduce the Marginal Fiscal Efficiency Index to partition tariffs into welfare-improving, trade-off, and revenue-decreasing regions. Expanding the trade war to more partners raises peak tariff revenue regardless of retaliation, whereas coordinated retaliation sharply erodes welfare. By January 2026, 20 percent of U.S. tariff rates exceed their Laffer peaks. Inverse-optimum estimation reveals diminished U.S. concern for foreign welfare, punitive treatment of China, and rising revenue motives.
\newline
\textbf{JEL\ Codes}: F11, F13, F14, F17
\end{abstract}

\newpage
 \pagenumbering{arabic}

\section{Introduction}
\label{sec:introduction}

In 2025, the United States launched a global trade war, raising import tariffs on most major trading partners. The costs of such policies are well understood. Tariffs distort relative prices (\citealt{FajgelbaumGoldberg2024}), protect inefficient domestic producers (\citealt{Trefler2004}), disrupt global value chains (\citealt{GrossmanHelpman2024}), and increase the cost of imported products (\citealt{CavalloGopinath2021}). Yet tariffs also raise government revenue (\citealt{santacreu2023tariffs}; \citealt{Lashkaripour2021}), with the cost of revenue generation partially borne by foreign countries (\citealt{johnsonOptimumTariffsRetaliation1953}; \citealt{pujorossbach2024}). For a policymaker focused on fiscal objectives, the central question is how far tariff increases can raise revenues before the contraction of the import base begins to erode them (\citealt{amitiImpact2018Tariffs2019}).

 The relationship between tariff revenues and tariff rates can be displayed intuitively through the tariff Laffer curve.  The popularization of the Laffer curve dates back to \citet{lafferwaniski}, and the concept has been studied widely within public finance \citep{harberger1964waste,trabandtuhlig2011}. Despite its intuitive appeal, historically it was rarely applied to trade policy, with recent exceptions including \citet{alessandriamixfriends} and \citet{HEAD2024103911}.  Classical trade theory emphasizes distinct motives behind tariffs outside of revenue, including improving the terms of trade \citep{johnsonOptimumTariffsRetaliation1953}, protection for sale \citep{grossmanhelpman}, and political effectiveness \citep{Fajgelbaum2019}. 

We complement this work by structurally quantifying tariff Laffer curves and their associated welfare trade-offs using the multi-sector Ricardian general equilibrium model developed by \citet{caliendoEstimatesTradeWelfare2015}.  We calibrate the model to study the 2025 U.S.\ trade war and compute tariff Laffer curves for the United States across trading partners and retaliation scenarios and formalize the relationship between marginal changes in welfare and revenue through an index we call the Marginal Fiscal Efficiency Index (MFEI). The MFEI varies from $-1$ to $1$ and partitions the tariff schedule into qualitatively distinct regions corresponding to below, between, and beyond the welfare and Laffer peaks.

Our analysis connects to a rapidly growing literature evaluating the macroeconomic consequences of recent global trade wars. While \citet{Fajgelbaum2019}, \citet{CavalloGopinath2021}, and \citet{amitiImpact2018Tariffs2019} document the price and welfare effects of the 2018--2019 U.S.-China trade war, a newer wave of research examines the sweeping 2025 escalation. For instance, \citet{NBERw33792} evaluate the dynamic impacts of the 2025 trade war across U.S.\ states, focusing on long-run capital accumulation and sectoral reallocation, while \citet{alessandriamixfriends} explore the role of tariffs as a revenue-generating tax cut. We complement this work by systematically mapping the static fiscal limits of protectionism via the tariff Laffer curve.

Quantitative studies have measured the magnitude of optimal tariffs and their welfare effects \citep{broda2008optimal,costinot2015comparative,pujorossbach2024}. A central benchmark in this literature is \citet{Ossa2014}, who finds optimal tariffs greater than 60 percent.  We find bilateral revenue-maximizing tariff rates for the United States between 20 and 30 percent, and welfare-maximizing rates between 0 and 10 percent, varying by trading partner and the partner's response.  Our lower rates are partly due to differences in trade elasticities, but mostly due to our full incorporation of general equilibrium and third-country effects.\footnote{If we consider bilateral changes in isolation, we reproduce both significantly higher peak rates and non-responsiveness of these rates to retaliation. See Supplemental Appendix \ref{app:robustness_elasticities}.} At peak rates, tariff revenue increases monotonically as the U.S.\ targets additional partners, even with retaliation.  Conversely, U.S.\ welfare increases only without retaliation; with retaliation, the U.S.\ experiences escalating welfare losses as the number of trading partners increases.  

We apply these findings to the tariff changes observed over the course of the 2025 trade war, using a new dataset of bilateral product-level tariff actions from January 2025 through January 2026 compiled from the WTO-IMF Tariff Tracker. The average U.S.\ import (export) tariff rose from 2.0 (6.0) percent in 2016 to 19.5 (9.0) percent by January 1, 2026, implying an additional \$364.0 billion in U.S.\ tariff revenue, \$110.0 billion in U.S.\ welfare losses, and \$394.1 billion in foreign welfare losses. 

We use an inverse-optimum (revealed-preference) approach to rationalize the observed tariff schedule. Similar approaches were used in the context of tariffs to estimate the relative weights on welfare versus political lobbying contributions \citep{goldbergmaggi1999,gawande2000} and in public finance to infer government preferences over redistributive taxation \citep{Saez2001,lockwood2017}. Complementing \citet{NBERw31798} and \citet{NBERw34658}, we use the variation in 2025 tariffs across partner-sector pairs to separately identify the marginal value placed on tariff revenue and the weights placed on partner welfare. Empirically, we find an implied weight on Chinese welfare turns sharply negative during 2025, consistent with punitive tariff-setting not explained by domestic welfare or revenue considerations alone.

\section{Model}
\label{sec:model}

We quantify the welfare and tariff-revenue effects of changes in tariff rates using the multi-country, multi-sector Ricardian trade model with input-output linkages from \citet{caliendoEstimatesTradeWelfare2015}. 

All markets are perfectly competitive, and households have Cobb-Douglas preferences over composite sectoral final goods, which aggregate a continuum of varieties within each sector via a constant elasticity of substitution (CES) technology. Each variety is produced using a Cobb-Douglas technology combining labor and composite intermediate inputs from all sectors, where composite intermediate inputs also aggregate varieties via the same CES technology. Each destination sources each variety from the lowest-cost supplier, where unit costs depend on wages, productivities, bilateral iceberg trade costs, and bilateral ad-valorem tariff rates. 

Productivities for each variety follow a Fr\'echet distribution as in \citet{eatonTechnologyGeographyTrade2002}, with country-sector-specific location parameters governing absolute advantage and sector-specific shape parameters equal to the trade elasticity \citep{SIMONOVSKA201434}. The combination of CES aggregation and Fr\'echet productivities implies that bilateral expenditure shares absorb the underlying productivity parameters, so that counterfactual analysis requires only observed trade shares and elasticities. We work in exact hat algebra \citep{dekle2008global}, which expresses counterfactual equilibria in proportional changes and avoids estimating unobserved productivities or iceberg trade costs.

There are $N$ countries indexed by $i,j$ and $S$ sectors indexed by $s,k$. Households in country $i$ have Cobb-Douglas preferences over sectoral final goods with expenditure shares $\{a_i^s\}_{s=1}^S$ satisfying $\sum_s a_i^s=1$. Production in sector $s$ and country $i$ combines labor and composite intermediate inputs with a Cobb-Douglas technology, where $\gamma_i^s$ is the value-added share and $\gamma_i^{k,s}$ is the intermediate input share from sector $k$ used in sector $s$, with $\gamma_i^s+\sum_{k=1}^{S}\gamma_i^{k,s}=1$. Shipping from exporter $j$ to importer $i$ in sector $s$ incurs an iceberg cost $\delta_{ij}^s$ and an ad-valorem tariff factor $\tau_{ij}^s\ge 1$, where $\tau_{ij}^s=1.25$ corresponds to a 25 percent tariff rate.

Let $X_{ij}^s$ denote importer $i$'s expenditure on sector-$s$ goods from origin $j$ at tariff-inclusive purchaser prices, and define total sectoral absorption $X_i^s\equiv \sum_{j=1}^N X_{ij}^s$. Bilateral expenditure shares are
\begin{equation}
\pi_{ij}^s \equiv \frac{X_{ij}^s}{X_i^s}.
\label{eq:shares}
\end{equation}

To compute counterfactual equilibria, we work in exact proportional changes. For any variable $x$, let $\hat{x}\equiv x'/x$ denote the ratio of its counterfactual value $x'$ to its baseline value $x$. Counterfactuals are indexed by tariff changes $\{\hat{\tau}_{ij}^s\}$, so that $\tau_{ij}^{s\prime}=\hat{\tau}_{ij}^s \tau_{ij}^s$. We hold iceberg costs fixed ($\hat{\delta}_{ij}^s=1$) and hold fundamental productivities fixed. The changes in unit costs and sectoral price indices satisfy
\begin{equation}
\hat{c}_i^s = \hat{w}_i^{\gamma_i^s} \prod_{k=1}^{S} (\hat{P}_i^k)^{\gamma_i^{k,s}}, \qquad \hat{P}_i^s = \left[ \sum_{j=1}^{N} \pi_{ij}^s \bigl( \hat{c}_j^s \hat{\tau}_{ij}^s \bigr)^{-\theta^s} \right]^{-1/\theta^s},
\label{eq:hat_prices}
\end{equation}
where $\theta^s$ is the sectoral trade elasticity. Counterfactual trade shares follow directly:
\begin{equation}
\pi_{ij}^{s\prime} = \frac{\pi_{ij}^s \bigl( \hat{c}_j^s \hat{\tau}_{ij}^s \bigr)^{-\theta^s}}{\sum_{m=1}^{N} \pi_{im}^s \bigl( \hat{c}_m^s \hat{\tau}_{im}^s \bigr)^{-\theta^s}}.
\label{eq:hat_trade}
\end{equation}
Note that baseline shares $\{\pi_{ij}^s\}$ absorb technology levels and other cost shifters that are not observed in the data.

General equilibrium imposes goods-market clearing, input-output accounting, factor-market clearing, and a country budget constraint. Let $X_i^{s\prime}$ denote counterfactual sector-$s$ absorption in country $i$ at purchaser prices, and let $Y_j^{s\prime}$ denote counterfactual sector-$s$ gross output in country $j$. Goods-market clearing requires
\begin{equation}
Y_j^{s\prime} = \sum_{i=1}^{N} \frac{\pi_{ij}^{s\prime} X_i^{s\prime}}{\tau_{ij}^{s\prime}},
\label{eq:goods_mkt}
\end{equation}
where dividing by $\tau_{ij}^{s\prime}$ converts tariff-inclusive expenditure into tariff-exclusive revenue received by the exporter. Sectoral absorption decomposes into final and intermediate demand,
\begin{equation}
X_i^{k\prime} = a_i^k I_i' + \sum_{s=1}^{S} \gamma_i^{k,s} Y_i^{s\prime},
\label{eq:absorption}
\end{equation}
and labor-market clearing implies
\begin{equation}
w_i' L_i = \sum_{s=1}^{S} \gamma_i^s Y_i^{s\prime}.
\label{eq:factor_mkt}
\end{equation}
Tariff revenue rebated to households in country $i$ is
\begin{equation}
T_i' = \sum_{s=1}^{S} \sum_{j=1}^{N} \left(\frac{\tau_{ij}^{s\prime}-1}{\tau_{ij}^{s\prime}}\right) \pi_{ij}^{s\prime} X_i^{s\prime},
\label{eq:tariff_revenue}
\end{equation}
and total household income is
\begin{equation}
I_i' = w_i' L_i + T_i' + D_i,
\label{eq:income}
\end{equation}
where $D_i$ is an exogenous nominal transfer that pins down the aggregate trade imbalance, held fixed in units of the num\'eraire.

We solve for wage changes $\{\hat{w}_i\}$ such that \eqref{eq:hat_prices}--\eqref{eq:income} hold jointly, normalizing the U.S.\ wage as the num\'eraire. We solve the system by fixed-point iteration on counterfactual wage changes until labor markets clear in all countries. Welfare is measured as the change in real income,
\begin{equation}
\hat{W}_i = \frac{\hat{I}_i}{\hat{P}_i}\;, \; \text{where} \; \hat{P}_i \equiv \prod_{s=1}^{S} (\hat{P}_i^s)^{a_i^s}.
\label{eq:welfare}
\end{equation}

\section{Data and Calibration}
\label{sec:calibration}

We assemble a novel dataset of bilateral applied tariff rates for the 2025 U.S.\ trade war by combining three sources. Our primary source is the WTO-IMF Tariff Tracker, which records implemented tariff actions and baseline (January 1, 2025) duties at the bilateral HS 6-digit product-code level.\footnote{\url{https://ttd.wto.org/en/analysis/tariff-actions}.} The tracker provides import tariff rates for all partners of each reporting country, but does not include the tariffs imposed on the exports of reporting countries by non-reporting partner countries. The database includes implemented tariff actions that have been confirmed by the affected importer and exporter, and does not include threatened or announced tariff actions that were never implemented.  We take baseline pre-trade-war tariff rates for non-reporting pairs from the 2019 revision of the CEPII MacMap-HS6 database \citep{MAcMap}. The MacMap database does not include contingent or retaliatory measures from the 2018 U.S.\ trade war; we therefore add the tariff rate increases for U.S.\ imports and exports assembled by \citet{Fajgelbaum2019}. Combining these sources yields bilateral applied ad-valorem tariff rates for each HS6 product on January 1, 2025 and a sequence of implemented changes through January 2026.\footnote{In Supplemental Appendix \ref{app-sec:GTD}, we verify that our results are robust to using the alternative Global Tariff Database (GTD) used by \citet{NBERw33792}, which provides hand-coded tariff rate increases for a subset of the 2025 timeline.}

We calibrate the model to the pre-trade-war baseline using 2023 data from the Eora multi-region input-output tables \citep{Lenzen2012,Lenzen2013}.\footnote{\url{https://worldmrio.com/eora26/}.} We aggregate the Eora data to obtain bilateral trade flows and sectoral input-output linkages across ten traded goods-producing sectors and five non-traded service sectors for 16 regions plus a Rest-of-World composite.\footnote{The regions are: United States, Canada, Mexico, China, European Union (EU-27), United Kingdom, Japan, South Korea, Australia, Turkey, Switzerland, Norway, Brazil, India, Indonesia, and Vietnam.}\textsuperscript{,}\footnote{Traded goods sectors are: Agriculture and Fishing; Mining; Food and Beverages; Textiles; Wood and Paper; Petroleum and Chemicals; Metals; Electrical and Machinery; Transport Equipment; and Other Manufacturing including Recycling. Service sectors are: Utilities and Construction; Trade and Hospitality; Transport and Communication; Finance and Business Services; and Public and Other Services.} For each policy date $d$ between January 1, 2025 and January 1, 2026, we aggregate bilateral HS6 applied tariff factors to Eora sectors using simple averages to obtain $\tau_{ij}^s(d)$, and define the tariff shock as $\hat{\tau}_{ij}^s(d) \equiv \tau_{ij}^s(d) / \tau_{ij}^s(\text{Jan 1, 2025})$.  

We calibrate preferences and production shares directly from the input-output accounts. Let $F_i^s$ denote final absorption of sector $s$ in country $i$, let $M_{ki}^s$ denote intermediate purchases from sector $k$ used in sector $s$ in country $i$, and let $VA_i^s$ denote sector-$s$ value added. Defining gross output as $Y_i^s \equiv VA_i^s + \sum_{k=1}^{S} M_{ki}^s$, we set
\begin{equation}
a_i^s = \frac{F_i^s}{\sum_{r=1}^{S} F_i^r}, \qquad \gamma_i^s = \frac{VA_i^s}{Y_i^s}, \qquad \gamma_i^{k,s} = \frac{M_{ki}^s}{Y_i^s}.
\label{eq:calib_shares}
\end{equation}
We calibrate $w_i L_i$ to match total value added for each region. We take traded-goods trade elasticities $\{\theta^s\}$ from \citet{Chen2023Fragmentation} and set $\theta^s = 8.35$ for non-traded services.

Eora reports flows at basic prices (net of tariffs). We construct purchaser-price expenditures by applying baseline tariff factors to the observed basic-price flows $\tilde{X}_{ij}^s$:
\begin{equation}
X_{ij}^s = 
\begin{cases}
\tau_{ij}^s \, \tilde{X}_{ij}^s & \text{if } j \neq i, \\
\tilde{X}_{ii}^s & \text{if } j = i,
\end{cases}
\qquad
\pi_{ij}^s = \frac{X_{ij}^s}{\sum_{m=1}^{N} X_{im}^s}.
\label{eq:purchaser_price}
\end{equation}
We compute each country's aggregate trade imbalance as net imports valued at tariff-exclusive prices,
\begin{equation}
D_i = \sum_{s=1}^{S} \left( \sum_{j=1}^{N} \frac{\pi_{ij}^s X_i^s}{\tau_{ij}^s} - \sum_{j=1}^{N} \frac{\pi_{ji}^s X_j^s}{\tau_{ji}^s} \right),
\label{eq:deficit}
\end{equation}
and hold the aggregate trade balance $D_i$ fixed in nominal units of the num\'eraire country (the U.S.) wage in all counterfactuals.

\section{Quantitative U.S.\ Tariff Laffer Curves}
\label{sec:laffer}

We compute tariff Laffer curves and associated welfare curves for the United States, first in bilateral disputes and then as the scope of the trade war expands to multiple partners.  For all exercises, we interpret each tariff schedule as a long-run counterfactual held fixed.  Dollar amounts are annual changes in real income or tariff revenue, measured in 2025 USD.\footnote{The Eora database is in billions of 2023 USD.  We scale all numbers to 2025 USD by multiplying by the mean change in U.S.\ current price GDP between 2023 and 2025, \url{https://fred.stlouisfed.org/series/GDP}.}  We report the ad-valorem tariff rate, $t$, for each corresponding tariff factor,  $\tau \equiv 1 + t$.

\subsection{U.S.\ Bilateral Tariff Laffer Curves}
\label{sec:bilateral_laffer}
We calibrate the model to the January 1, 2025 tariff and trade-flow baseline (Section~\ref{sec:calibration}). For each partner $j$, we impose a uniform U.S.\ ad-valorem tariff rate $\tau$ on imports from $j$ in each of the ten traded-goods sectors, varying $t$ from 0 to 50 percent in one-percentage-point increments. All other bilateral tariffs are held fixed at their January 1, 2025 values.

To isolate the effect of the tariff level from heterogeneity in baseline bilateral duties, we report changes in U.S.\ tariff revenue and welfare relative to a partner-specific reference equilibrium with bilateral free trade between the United States and $j$ (i.e., $\tau_{US,j}^s=\tau_{j,US}^s=1$ for traded-goods sectors), holding all other tariff wedges at the January 1, 2025 baseline. Revenue and welfare changes are measured in billions of U.S.\ dollars, and include third-country effects.

For each U.S.\ tariff rate, $t$, we consider four partner responses: no response; equivalent response (the partner matches the U.S.\ tariff); revenue-maximizing response; and welfare-maximizing response. For the revenue- and welfare-maximizing responses, we compute the partner’s optimal uniform tariff on U.S.\ exports by searching over $\{0,0.1,\dots,50\}$ percent. As we step through U.S.\ tariffs, we initialize the search at the previous optimum and stop once the partner objective declines for both a $+0.1$ and a $-0.1$ deviation. This yields a fast, accurate approximation to the partner best-response correspondence along the policy grid.

Figure~\ref{fig:usa_laffer_eu_chn} shows the tariff Laffer curves (Panels a and c) and welfare curves (Panels b and d) for the European Union and China, respectively. For both partners, the U.S.\ revenue-maximizing tariff is in the mid-twenties and is only modestly affected by how the partner responds. Welfare peaks at much lower tariff rates and declines rapidly beyond the peak. A comparison across the panels reveals that retaliation reduces peak welfare more than it does peak revenue, with equivalent retaliation eliminating nearly all welfare gains while reducing peak U.S.\ revenue only modestly. The curves for China are qualitatively similar, but quantitatively distinctive in that the potential welfare changes are substantially greater for a U.S.\ bilateral trade war with China.

The figures also show that if the European Union and China were to implement bilateral welfare-maximizing or revenue-maximizing tariffs on imports from the United States, this would lead to welfare losses for the U.S.\ relative to the baseline bilateral free-trade equilibrium. These losses are partially offset if the U.S.\ raises tariffs in response.\footnote{Supplemental Appendix \ref{app:additional_laffers_baselines} shows the U.S.\ experiences small welfare gains relative to January 1, 2025 baseline tariff rates; consistent with \citet{pujorossbach2024}.}  

\begin{figure}[htbp]
    \centering
    \caption{U.S.\ Laffer and Welfare Curves: European Union and China \\ \emph{\small{Changes Relative to U.S.\ Bilateral Free Trade Baseline}}}
    \label{fig:usa_laffer_eu_chn}
    \vspace{0.5cm}
    \includegraphics[width=\textwidth]{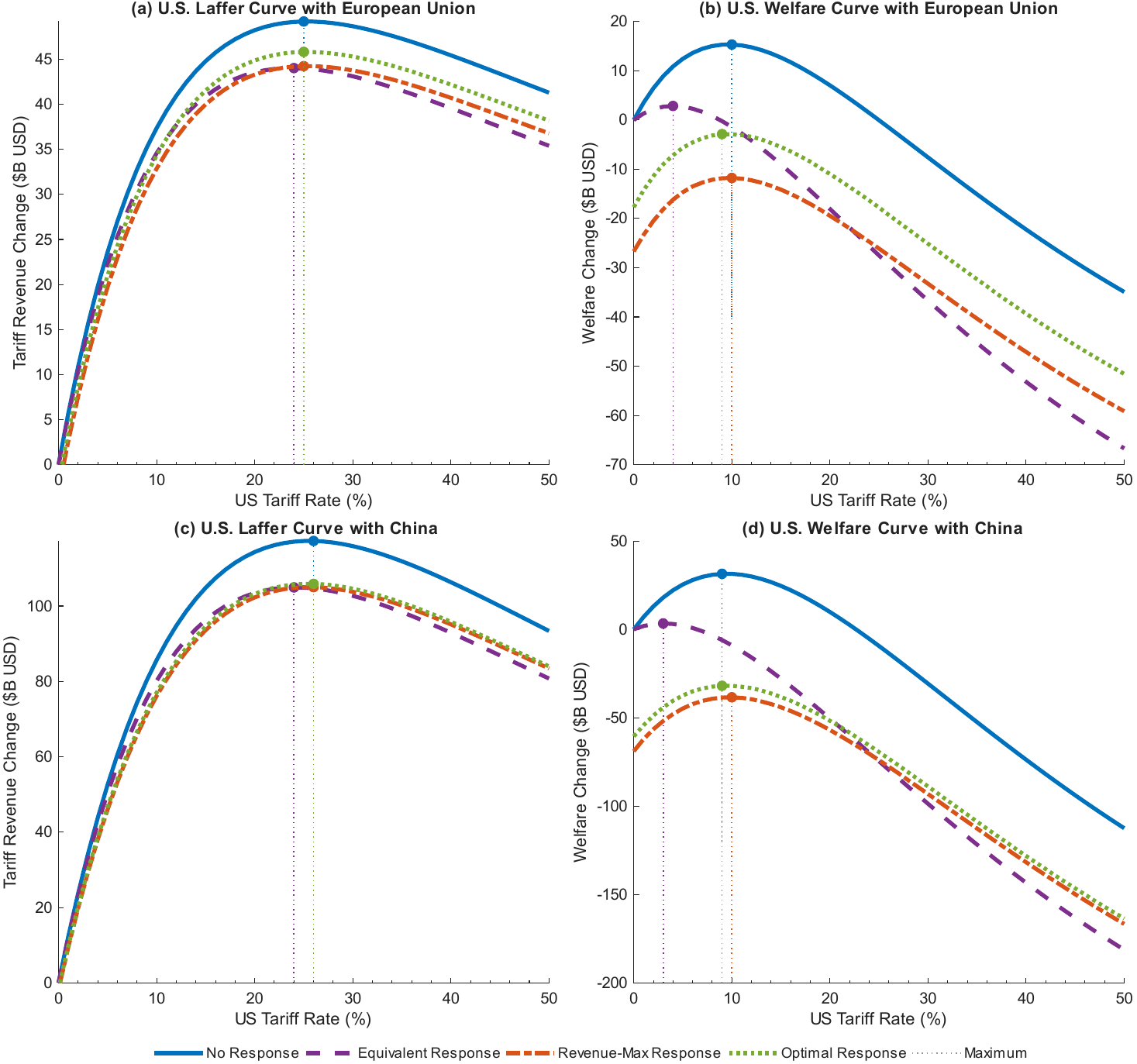}
\end{figure}

At the bilateral revenue peaks, welfare is already negative for both partners. Under no response, the E.U. Laffer peak occurs at a 25 percent tariff and yields \$51.7 billion in revenue and a near-zero change in U.S.\ welfare; the China Laffer peak occurs at 26 percent and yields \$123.2 billion in revenue but \textminus\$14.3 billion in U.S.\ welfare. Welfare-maximizing tariffs are much lower: 9 percent for China and 10 percent for the E.U., delivering gains of \$33.0 billion and \$16.1 billion, at their respective rates. Under equivalent retaliation, peak revenue falls but remains substantial (e.g., \$46.3 billion for the E.U. and \$110.3 billion for China), while welfare is near zero at tariffs under 5 percent and turns sharply negative as rates increase.

Across the 15 largest U.S.\ trading partners, bilateral welfare-maximizing tariffs cluster in the high single digits under no response, averaging 9.5 percent and ranging from 9 to 11 percent. Bilateral revenue-maximizing tariffs are higher, ranging from 21 to 31 percent and averaging 25.5 percent. At the bilateral Laffer peaks, partner retaliation is quantitatively more important for welfare than for revenue: under equivalent retaliation, peak U.S.\ tariff revenue falls from \$292.2 billion to \$257.8 billion, but U.S.\ welfare at those peaks becomes sharply negative for every partner.

\subsection{Marginal Excess Burden and the Marginal Fiscal Efficiency Index}
\label{sec:MEB}

The disparity between the welfare peak and the Laffer peak has a natural interpretation in the language of public finance. The marginal excess burden (MEB) measures the welfare cost per dollar of revenue raised by a marginal increase in a tax instrument \citep{harberger1964waste, ballard1985marginal, browning1987marginal}. We compute the tariff-policy analog:
\begin{equation}
\text{MEB}(\tau) \equiv -\frac{dW/d\tau}{dR/d\tau},
\label{eq:MEB}
\end{equation}
where $W$ denotes U.S.\ welfare (real income, including the lump-sum rebate of tariff revenue), $R$ denotes U.S.\ tariff revenue, and $\tau$ denotes the ad-valorem tariff rate. 

In the standard domestic public finance setting, the MEB is typically positive: raising an additional dollar of revenue destroys more than zero dollars in welfare because of deadweight loss.\footnote{The possibility of a negative MEB has precedent in the domestic tax literature where tax instruments raise revenue while correcting a preexisting distortion.  For example, \citet{fullerton1991reconciling} considers an investment tax credit that corrects interasset distortions and \citet{Goulder1995} survey the ``double dividend'' hypothesis under which revenue-neutral environmental taxes incur gross benefits.} Tariffs in an international setting differ in an important respect. The welfare change $dW/d\tau$ reflects not only deadweight loss from domestic distortions but also terms-of-trade effects: a tariff that depresses the world price of imports transfers welfare from foreign exporters to domestic consumers \citep{BagwellGATT}. At low tariff rates, the terms-of-trade gain can exceed the deadweight loss, so the MEB is negative with domestic welfare rising alongside revenue \citep{johnsonOptimumTariffsRetaliation1953, pujorossbach2024}.

To evaluate whether a tariff is fiscally attractive, we compare $\text{MEB}_{\text{tariff}}(\tau)$ to the marginal excess burden of raising a dollar through the domestic tax system, denoted $\alpha$. A tariff is fiscally efficient (from a domestic standpoint) when $\text{MEB}_{\text{tariff}}(\tau)=\alpha$; increasing tariffs is fiscally desirable when $\text{MEB}_{\text{tariff}}(\tau)<\alpha$ and fiscally undesirable when $\text{MEB}_{\text{tariff}}(\tau)>\alpha$ \citep{Dahlby2008}. 

We summarize this comparison using the Marginal Fiscal Efficiency Index (MFEI):
\begin{equation}
\text{MFEI}(\tau) \equiv \frac{dW/d\tau + \alpha \cdot dR/d\tau}{\,|dW/d\tau| + \alpha \cdot |dR/d\tau|\,}.
\label{eq:MFEI}
\end{equation}
The MFEI partitions tariff rates into three regions. When $dW/d\tau>0$ and $dR/d\tau>0$, $\text{MFEI}=1$ (a ``Free-Lunch'' region). When $dW/d\tau<0$ and $dR/d\tau<0$, $\text{MFEI}=-1$ (a ``Beyond-Laffer'' region). In between, $dW/d\tau<0<dR/d\tau$, $\text{MFEI}\in\left(-1,1\right)$ (a ``Trade-Off'' region).\footnote{In Supplemental Appendix \ref{app-sec:MFEI} we further decompose the ``Trade-Off'' region into fiscally-efficient and fiscally-inefficient trade-offs depending on $\alpha$.}

We set $\alpha=0.25$, following OMB Circular A-94 \citep{omb2023a94}. Figure~\ref{fig:MEB_no_response} plots the implied MEB and MFEI across partners. Under no response, tariffs begin in a Free-Lunch zone, enter a Trade-Off zone as welfare turns down, and eventually move beyond the Laffer peak. This three-region structure is qualitatively robust across response scenarios. Under equivalent retaliation, the Free-Lunch zone shrinks sharply, while the Laffer peak shifts only modestly lower, widening the range of tariff rates over which the U.S.\ government faces a genuine trade-off between revenue and welfare. The MFEI therefore provides a compact diagnostic of whether a given tariff lies in a welfare-improving, trade-off, or revenue-decreasing region.

\begin{figure}[htbp]
    \centering
    \caption{Marginal Welfare Cost Measures of U.S.\ Tariff Revenue \\ \emph{\small{Changes Relative to U.S.\ Bilateral Free Trade Baseline}}}
    \label{fig:MEB_no_response}
    \includegraphics[width=\textwidth]{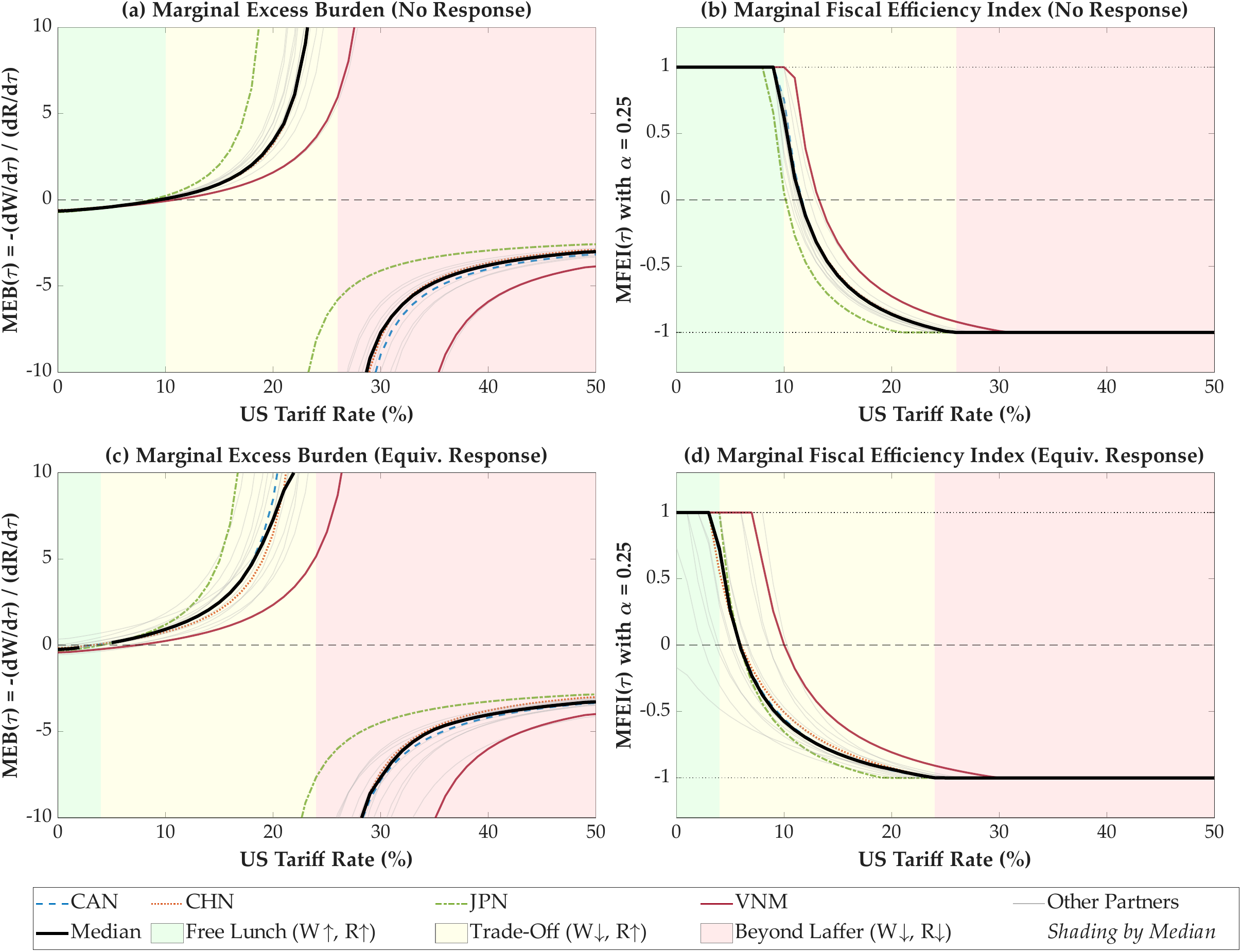}
\end{figure}

\subsection{Laffer Curves for a Trade War with Multiple Partners}
\label{sec:cumulative}

We next vary a uniform U.S.\ tariff against an expanding set of partners. We start with China and sequentially add partners (Canada, Mexico, then others ordered by their bilateral MEB at $\tau=0$; see the figure notes for the ordering).\footnote{Supplemental Appendix \ref{sec-app:cumulative_opt_response} shows the results are qualitatively similar with optimal partner responses, which we approximate via a single round of sequential optimization.}

Figure~\ref{fig:cum_Laffer_4panel} displays cumulative Laffer and welfare curves under two response scenarios: no response (panels a--b) and equivalent response (panels c--d) with each partner matching the U.S.\ tariff rate on its imports from the U.S.

\begin{figure}[htbp]
    \centering
    \caption{U.S.\ Laffer and Welfare Curves: Multiple Trading Partners \\ \emph{\small{Changes Relative to U.S.\ Free Trade Baseline with All Partners}}}
    \label{fig:cum_Laffer_4panel}
    \includegraphics[width=\textwidth]{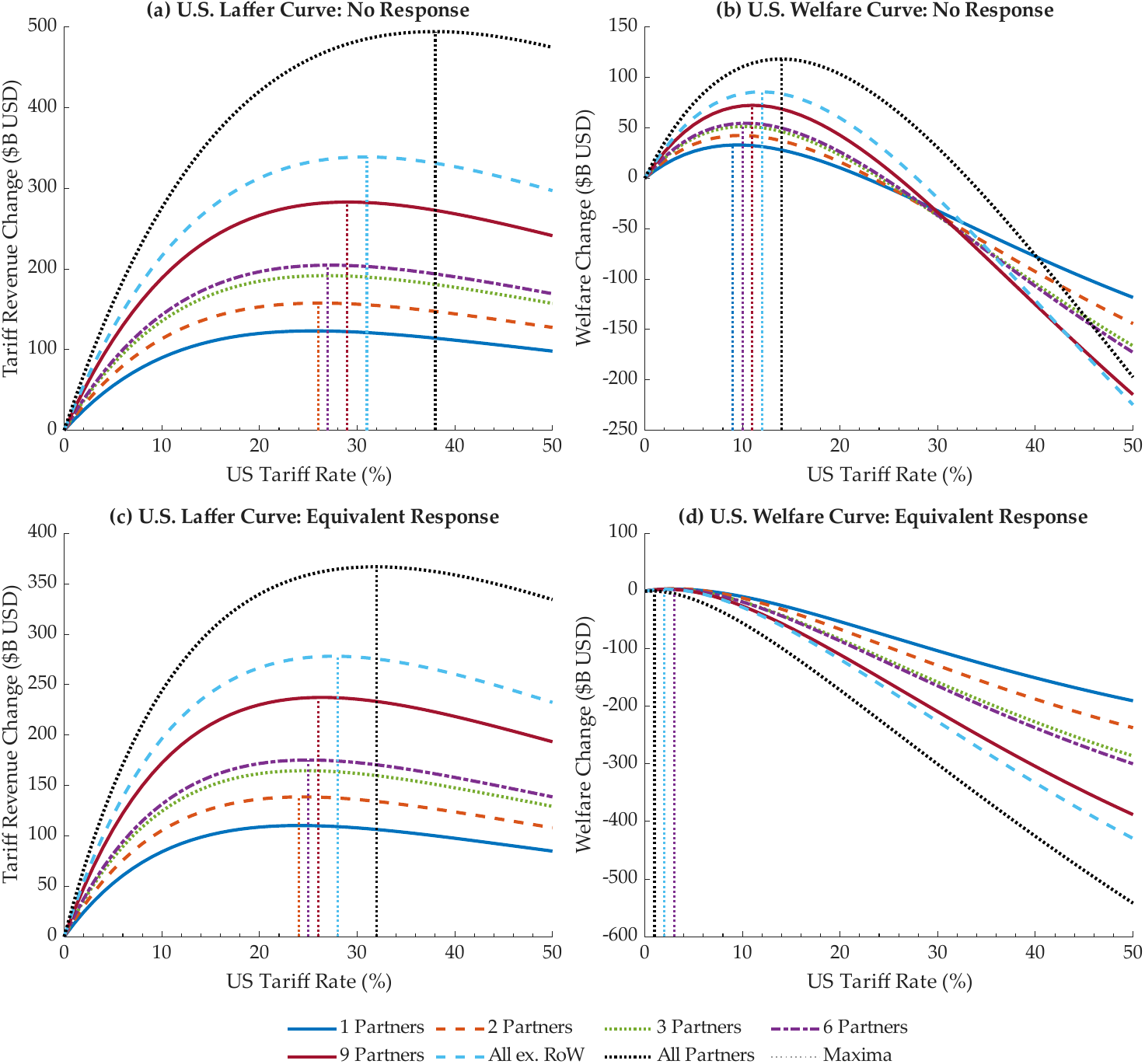}
    \vspace{0.5em}
    \begin{minipage}{\textwidth}
        \begin{singlespacing}
            \fontsize{8.5}{10}\selectfont
            \noindent\textit{Notes:} ``\# Partners'' indicates the number of partners on which the U.S.\ imposes tariffs and from which it potentially faces retaliation. Partner order: (1) CHN (2) CAN (3) MEX (4) VNM (5) TUR (6) KOR (7) IND (8) IDN (9) EUU (10) BRA (11) NOR (12) CHE (13) JPN (14) GBR (15) AUS (16) Rest of World.   Results robust to alternative orderings.
        \end{singlespacing}
    \end{minipage}
\end{figure}

Two patterns stand out from comparing the revenue and welfare panels. First, as shown in Panels (a) and (c), peak U.S.\ tariff revenue rises monotonically with the number of partners. In the no-response scenario (Panel a), the uniform-tariff Laffer peak shifts up as the U.S.\ expands the set of targeted partners, and peak revenue rises from \$123.2 billion in a bilateral dispute with China to \$494.4 billion in a global trade war. 

Second, Panels (b) and (d) demonstrate that welfare is far more sensitive to retaliation than revenue. Without retaliation (Panel b), a global trade war can generate positive U.S.\ welfare at low tariff rates through terms-of-trade gains, but welfare turns negative well before the revenue peak. Under equivalent retaliation (Panels c and d), the revenue implications are modestly attenuated, falling from \$494.4 billion to \$367.0 billion.  The welfare gains, however, vanish entirely with retaliation, with U.S.\ welfare losses for essentially all positive tariff rates (Panel d) and losses scaling with the number of retaliating trading partners. In a global trade war under equivalent retaliation, the welfare-maximizing tariff remains close to 0 and U.S.\ welfare losses exceed \$250 billion at the Laffer peak.

This asymmetry between the revenue and welfare effects of retaliation has a direct policy implication. A government that places substantial weight on tariff revenue faces strong incentives to escalate a trade war to additional partners regardless of whether those partners retaliate, because retaliation reduces peak revenue only modestly. A government that places weight primarily on welfare faces the opposite incentive: retaliation eliminates all welfare gains and generates losses that scale with the number of partners involved.

\section{The 2025 U.S.\ Trade War}
\label{sec:tradewar}

We now apply the framework to the sequence of tariff changes observed over the course of the 2025 U.S.\ trade war. Because tariffs changed frequently, and many measures were temporary, the model should be interpreted as describing the long-run impact if a given tariff schedule were held in place.

To facilitate comparison with the 2018 U.S.-China trade war and other benchmarks, we report counterfactual changes relative to 2016 applied tariffs from the MacMap database. We calibrate the model to January 1, 2025 quantities and compute the equilibrium implied by moving from January 1, 2025 tariffs to 2016 tariffs; we take this equilibrium as the 2016 baseline. For each policy date $d$, we then solve for the equilibrium under tariffs $\tau_{ij}^s(d)$ and report changes relative to 2016.

Table~\ref{tbl:us_tariff_timeline} reports average U.S.\ tariffs on imports and exports at each major policy date, together with model-implied changes in U.S.\ tariff revenue, U.S.\ welfare, and rest-of-world welfare. For each date, we report the effects of U.S.\ tariffs alone and the incremental effect of observed retaliation.

\begin{table}[h!]
\centering
\caption{Impact of U.S. Tariffs and Retaliation, Changes from Pre-Trade War Baseline}
\label{tbl:us_tariff_timeline}
\begin{tabular}{lrrrrrrrr}
\toprule
& \multicolumn{2}{c}{\textbf{Avg. USA}} & \multicolumn{2}{c}{\textbf{USA Tariff}} & \multicolumn{2}{c}{\textbf{USA }} & \multicolumn{2}{c}{\textbf{Rest of World}} \\
& \multicolumn{2}{c}{\textbf{Tariff Rate (\%)}} & \multicolumn{2}{c}{\textbf{Revenue ($\Delta$\$B)}} & \multicolumn{2}{c}{\textbf{Welfare ($\Delta$\$B)}} & \multicolumn{2}{c}{\textbf{Welfare ($\Delta$\$B)}} \\
\cmidrule(lr){2-3} \cmidrule(lr){4-5} \cmidrule(lr){6-7} \cmidrule(lr){8-9}
\textbf{Date} & \textbf{Imports} & \textbf{Exports} & \shortstack[c]{\textbf{U.S.} \\ \textbf{Only}} & \shortstack[c]{\textcolor{textgray}{\textbf{+/- w/}} \\ \textcolor{textgray}{\textbf{Retal.}}} & \shortstack[c]{\textbf{U.S.} \\ \textbf{Only}} & \shortstack[c]{\textcolor{textgray}{\textbf{+/- w/}} \\ \textcolor{textgray}{\textbf{Retal.}}} & \shortstack[c]{\textbf{U.S.} \\ \textbf{Only}} & \shortstack[c]{\textcolor{textgray}{\textbf{+/- w/}} \\ \textcolor{textgray}{\textbf{Retal.}}} \\
\cmidrule(lr){1-1} \cmidrule(lr){2-3} \cmidrule(lr){4-5} \cmidrule(lr){6-7} \cmidrule(lr){8-9}
\multicolumn{9}{l}{\rlap{\textit{2016 Baseline Tariffs (MacMAP) and 2019 Trade War Tariffs (MacMAP + \citealt{Fajgelbaum2019})}}} \\
2016 & 2.0 & 6.0 & 0.0 & \textcolor{textgray}{---} & 0.0 & \textcolor{textgray}{---} & 0.0 & \textcolor{textgray}{---} \\
2019 (w/Trade War) & 4.6 & 7.2 & 91.9 & \textcolor{textgray}{+1.6} & -28.3 & \textcolor{textgray}{+23.7} & -55.2 & \textcolor{textgray}{-1.8} \\
\addlinespace[0.5ex]
\midrule
\multicolumn{9}{l}{\textit{2025 Tariff Timeline (WTO-IMF  Tariff Tracker Database)}} \\
Jan 1 & 5.5 & 8.5 & 109.5 & \textcolor{textgray}{+0.0} & -46.9 & \textcolor{textgray}{+0.0} & -94.6 & \textcolor{textgray}{+0.0} \\
Mar 4 & 10.4 & 9.0 & 198.9 & \textcolor{textgray}{-1.2} & -129.7 & \textcolor{textgray}{-3.3} & -229.9 & \textcolor{textgray}{+2.8} \\
Apr 4 & 8.5 & 9.0 & 164.2 & \textcolor{textgray}{-1.2} & -111.6 & \textcolor{textgray}{-6.1} & -198.1 & \textcolor{textgray}{+4.3} \\
Apr 12 & 11.4 & 7.9 & 191.7 & \textcolor{textgray}{-8.0} & -286.4 & \textcolor{textgray}{-27.8} & -348.1 & \textcolor{textgray}{-61.9} \\
May 14 & 14.7 & 9.4 & 293.0 & \textcolor{textgray}{-6.0} & -116.3 & \textcolor{textgray}{-18.4} & -317.1 & \textcolor{textgray}{-4.9} \\
Aug 18 & 20.1 & 9.5 & 357.8 & \textcolor{textgray}{-7.3} & -136.6 & \textcolor{textgray}{-18.3} & -407.2 & \textcolor{textgray}{-4.4} \\
Jan 1 '26 & 19.5 & 9.0 & 369.1 & \textcolor{textgray}{-5.1} & -96.9 & \textcolor{textgray}{-13.1} & -383.2 & \textcolor{textgray}{-10.9} \\
\addlinespace[0.5ex]
\multicolumn{9}{l}{\textit{Free Trade Scenarios (Retaliation = Reciprocation)}} \\
USA & 0.0 & 0.0 & -88.6 & \textcolor{textgray}{+0.0} & -70.5 & \textcolor{textgray}{+158.0} & 86.2 & \textcolor{textgray}{-160.0} \\
Global & 0.0 & 0.0 & --- & \textcolor{textgray}{+0.0} & --- & \textcolor{textgray}{+139.1} & --- & \textcolor{textgray}{+117.4} \\
\addlinespace[0.5ex]
\multicolumn{9}{l}{\textit{Equivalent Response Scenarios: Increase Tariffs to Match Partner Rates}} \\
MacMAP 2016 & 5.6 & 6.0 & 125.9 & \textcolor{textgray}{+5.0} & -19.8 & \textcolor{textgray}{+44.4} & -68.7 & \textcolor{textgray}{-29.4} \\
Jan 1 & 8.8 & 8.6 & 216.9 & \textcolor{textgray}{-13.6} & 13.2 & \textcolor{textgray}{-50.8} & -195.2 & \textcolor{textgray}{+20.4} \\
\addlinespace[0.5ex]
\multicolumn{9}{l}{\textit{U.S. Welfare-Maximizing and Revenue-Maximizing Uniform Tariff Rates}} \\
Welfare, No Retal. & 15.7 & 8.8 & 374.0 & \textcolor{textgray}{---} & 20.7 & \textcolor{textgray}{---} & -268.4 & \textcolor{textgray}{---} \\
Welfare w/ Retal. & 5.6 & 8.7 & 109.5 & \textcolor{textgray}{-0.2} & -46.9 & \textcolor{textgray}{-2.1} & -94.6 & \textcolor{textgray}{+2.7} \\
Revenue, No Retal. & 38.0 & 8.9 & 516.2 & \textcolor{textgray}{---} & -152.9 & \textcolor{textgray}{---} & -561.3 & \textcolor{textgray}{---} \\
Revenue w/ Retal. & 32.0 & 32.0 & 511.5 & \textcolor{textgray}{-94.4} & -89.5 & \textcolor{textgray}{-190.8} & -496.9 & \textcolor{textgray}{+35.4} \\
\midrule\midrule
\multicolumn{9}{c}{\textbf{Welfare and Revenue Effects of January 1, 2026 Bilateral Tariffs vs 2016 Baseline}} \\
\midrule
\multicolumn{3}{l}{\textbf{Partner Name}} & \textbf{CHN} & \textbf{EUU} & \textbf{MEX} & \textbf{CAN} & \textbf{JPN} & \textbf{All} \\
\midrule
\multicolumn{3}{l}{US Tariff on Imports (\%)} & 42.2 & 15.8 & 8.6 & 7.3 & 18.6 & 19.5 \\
\multicolumn{3}{l}{$\Delta$ US Tariff Revenue (\$B)} & 92.4 & 146.2 & 136.7 & 135.0 & 119.7 & 364.0 \\
\multicolumn{3}{l}{$\Delta$ US Welfare (\$B)} & -159.0 & -46.4 & -40.7 & -42.1 & -48.2 & -110.0 \\
\multicolumn{3}{l}{$\Delta$ Partner Welfare (\$B)} & -262.3 & -29.1 & -6.1 & -14.5 & -8.1 & -394.2 \\
\addlinespace[0.5ex]
\midrule
\multicolumn{3}{l}{\shortstack[l]{US Welfare Cost \\\\per \$ of US Tariff Revenue}} & \$1.72 & \$0.32 & \$0.30 & \$0.31 & \$0.40 & \$0.30 \\
\midrule
\multicolumn{3}{l}{\shortstack[l]{Global Welfare Cost \\\\ per \$ of US Tariff Revenue}} & \$3.58 & \$1.18 & \$1.14 & \$1.11 & \$1.28 & \$1.39 \\
\bottomrule
\end{tabular}
\end{table}

By January 2026, the average U.S.\ import tariff rose from 2.0 percent (2016 baseline) to 19.5 percent and U.S.\ tariff revenue rose to \$364.0 billion ($= 369.1 - 5.1$). U.S.\ welfare is negative at every 2025 date, reaching \textminus\$96.9 billion (\textminus\$110.0 billion including observed retaliation). Retaliation has a modest effect on U.S.\ revenue but a larger effect on welfare, consistent with the asymmetry documented in Section~\ref{sec:cumulative}.

The table also reports uniform-tariff benchmarks. Under no retaliation, the welfare-maximizing uniform tariff is 15.7 percent and the revenue-maximizing uniform tariff is 38.0 percent.  The observed average tariff therefore lies between these benchmarks, but welfare is negative because the schedule is highly non-uniform across partners. The bilateral decomposition highlights China as an outlier: China’s average tariff is 42.2 percent---well above its bilateral Laffer peak---so reducing tariffs on China would raise both U.S.\ revenue and U.S.\ welfare. At the January 2026 schedule, each dollar of U.S.\ tariff revenue from China costs \$1.72 in U.S.\ welfare, compared to roughly \$0.30--\$0.40 for other major partners.  The disparity is even starker from a global perspective. The global welfare cost per dollar of U.S.\ tariff revenue from China is exceptionally high at \$3.58, whereas for other major partners it ranges from \$1.11 (Canada) to \$1.28 (Japan). Across all partners combined, each dollar of U.S.\ tariff revenue raised costs \$0.30 in U.S.\ welfare and \$1.39 in global welfare.

With equivalent retaliation, the welfare-maximizing uniform tariff falls from 15.7 to 5.6 percent (with welfare losses compared to 2016 tariffs) and the revenue-maximizing tariff remains elevated at 32 percent.  Notably, even with equivalent retaliation, revenue-maximizing tariffs produce \$417.1 billion in tariff revenue for the U.S.

As additional context, the free-trade benchmarks in the middle panel show that global free trade generates gains for both the United States and the rest of the world in the model, while unilateral U.S.\ free trade (holding foreign tariffs fixed) reduces U.S.\ welfare by eliminating the terms-of-trade component of existing tariff wedges. Average export tariffs are trade-weighted averages computed using the counterfactual trade shares implied by each scenario, so composition effects can shift the reported average even when statutory rates fall for some partners.

\section{Inferring the U.S.\ Government Objective Function}
\label{sec:inverse_optimum}

The results in Table~\ref{tbl:us_tariff_timeline} raise a natural question: what objective function rationalizes the observed tariff schedule?  Under the assumption that observed tariffs satisfy the first-order conditions of the government's optimization problem, we can recover the implied policy weights from the government. This strategy, referred to as an inverse-optimum approach, has been used in international trade to estimate the relative weight on welfare versus political lobbying contributions \citep{goldbergmaggi1999, gawande2000}, and in public finance to infer government preferences over redistributive taxation \citep{Saez2001, lockwood2017}, with recent examples including \citet{NBERw31798} and \citet{NBERw34658}.

\subsection{Government Objective Function and the Shadow Wedge}

We assume the U.S.\ government sets tariffs $\boldsymbol{\tau}_{US} \equiv \{\tau_{US,j}^s\}_{s \in S, j \neq US}$ to maximize
\begin{equation} \label{eq:gof}
G_{US}(\boldsymbol{\tau}_{US}) = W_{US}(\boldsymbol{\tau}_{US}) + \alpha \cdot T_{US}(\boldsymbol{\tau}_{US}) + \Omega_{US}(\boldsymbol{\tau}_{US}),
\end{equation}
where $W_{US}$ is U.S.\ welfare (consumption-equivalent variation in USD, including the lump-sum rebate of tariff revenue), $T_{US}$ is tariff revenue, and $\Omega_{US}(\cdot)$ is a shadow wedge that captures all other considerations. Because $W_{US}$ includes tariff revenue rebates to consumers, $\alpha$ should be interpreted as an additional marginal value of revenue beyond its contribution to domestic welfare, e.g., from using tariff revenue to relax a distortionary-tax financing constraint.

The first-order conditions imply
\begin{equation} \label{eq:foc}
\frac{\partial W_{US}}{\partial \tau_{US,j}^s} + \alpha \cdot \frac{\partial T_{US}}{\partial \tau_{US,j}^s} + \frac{\partial \Omega_{US}}{\partial \tau_{US,j}^s} = 0,
\end{equation}
and we define the marginal shadow wedge as
\begin{equation} \label{eq:omega_measured}
\omega_j^s \equiv \frac{\partial \Omega_{US}}{\partial \tau_{US,j}^s} = -\left( \frac{\partial W_{US}}{\partial \tau_{US,j}^s} + \alpha \cdot \frac{\partial T_{US}}{\partial \tau_{US,j}^s} \right).
\end{equation}
Under concavity, $\partial(W_{US}+\alpha T_{US})/\partial\tau>0$ implies tariffs are below the domestic economic optimum and $\partial(W_{US}+\alpha T_{US})/\partial\tau<0$ implies tariffs are above it. Because $\omega_j^s$ is the negative of this domestic marginal value, $\omega_j^s<0$ corresponds to restraint (forces pushing tariffs down relative to the domestic optimum) and $\omega_j^s>0$ corresponds to aggression (forces pushing tariffs up).

We compute $\omega_j^s$ by perturbing each tariff factor by $\Delta(\tau_{US,j}^s) = 0.001$ and solving for the full general-equilibrium response:
\begin{equation}
    \omega_j^s \approx  \frac{-\Delta\left(W_{US}\right) - \alpha \cdot \Delta\left(T_{US} \right)}{\Delta\left(\tau_{US,j}^s\right)}.
\end{equation}

To compute the shadow wedge, we must take a stand on $\alpha$. Estimating $\alpha$ jointly with the shadow wedge parameters is problematic when a large share of observed tariff-partner-sector combinations lie in either the Free-Lunch zone (where $\partial W_{US}/\partial \tau > 0$ and $\partial T_{US}/\partial \tau > 0$) or the Beyond-Laffer zone (where both derivatives are negative). An unconstrained estimation would attempt to rationalize tariffs in these zones by assigning a negative weight on tariff revenue, which is economically implausible. Attempting to bound or restrict the sample to the Trade-Off zone, where $\alpha$ is well identified, introduces selection effects across time as the share of observations in each zone shifts over the course of the trade war. 

We therefore adopt a straightforward approach and set $\alpha = 0.25$ based on OMB guidance \citep{omb2023a94}, consistent with Section~\ref{sec:MEB}. Fixing $\alpha$ avoids these unstable estimates and evaluates tariffs using the domestic economic component as a transparent benchmark, placing equivalent value on tariff revenue as revenue from any other federal tax instrument.

\subsection{Decomposing the Shadow Wedge}

We decompose the shadow wedge by assuming the U.S.\ government objective internalizes the effects of U.S.\ tariffs on foreign welfare.\footnote{For example: $\Omega_{US}(\boldsymbol{\tau}_{US}) = \gamma^W \cdot W_{-US}(\boldsymbol{\tau}_{US}) + \sum_{k \neq US} \gamma^{g(k)} \cdot W_k(\boldsymbol{\tau}_{US})$.} To estimate the extent of this internalization, we run the following regression:
\begin{equation}
\omega_{j,d}^s \;=\; \gamma_{r(d)}^W \left( \frac{\Delta W_{-US}}{\Delta \tau_{US,j}^s} \right) + \psi^{g(j)}_{r(d)} \left( \frac{\Delta W_j}{\Delta \tau_{US,j}^s} \right) + \kappa^{s}_{r(d)} + e_{j,d}^s,
\end{equation}
where each observation is a tariff perturbation indexed by sector $s$, partner $j$, and date $d$ during regime-period $r(d)$. $W_{-US}=\sum_{k\neq US} W_k$ is total non-U.S.\ welfare, and $g(j) \in \{\text{China, USMCA, Aligned, Not-Aligned}\}$ assigns partners to geopolitical groups using alignment scores from \citet{Baileyetal2017}.

The coefficient $\gamma^W_{r(d)}$ is the uniform weight placed on aggregate foreign welfare and $\psi^{g(j)}_{r(d)}$ is the targeted-partner's group-specific weight.  The double-counting is intentional and captures to what extent the U.S.\ internalizes the effects of tariffs on a targeted partner, beyond its effects on foreign welfare.  The sum $\gamma^W_{r(d)} + \psi^{g(j)}_{r(d)}$ represents the total effective weight the U.S.\ places on targeted-partner $j$'s welfare.   The $\kappa^{s}_{r(d)}$ are sector-by-regime fixed effects for $s \in \{\text{Metals, Transport}\}$, which receive differential tariff treatment in 2025; our results are robust to excluding these terms. Inference relies on 90\% confidence intervals from a wild cluster bootstrap (country-date-level clusters) using \citet{webb2023} weights.

Lastly, we estimate the average MEB, $\tilde{\alpha}_{r(d)}$, for tariffs in the Trade-Off zone, i.e., where $\frac{\Delta W_{US}}{\Delta \tau_{US,j}^s} < 0 < \frac{\Delta T_{US}}{\Delta \tau_{US,j}^s}$ as
\begin{equation}
\frac{-\Delta W_{US}}{\Delta \tau_{US,j}^s} = \tilde{\alpha}_{r(d)} \cdot \frac{\Delta T_{US}}{\Delta \tau_{US,j}^s} + u_{j,d}^s,
\end{equation}
where $u_{j,d}^s$ is an error term.  We do not use this estimate as a substitute for $\alpha=0.25$ as doing so would introduce time-varying selection-effects.\footnote{Supplemental appendix \ref{app-sec:inv_opt_alpha} shows our findings are robust to alternative values of $\alpha$.}

\subsection{Results}

Figure~\ref{fig:invopt_results} summarizes the inverse-optimum estimates across four tariff regime-periods, $r(d)\in\{\text{2010--2016; 2019; April--May 2025; August 2025--January 2026}\}$, systematically mapping the components of the government's implied objective function over time. 

\begin{figure}[htbp]
    \centering
    \caption{U.S.\ Inverse-Optimum Policy Weights}
    \label{fig:invopt_results}
    \includegraphics[width=\textwidth]{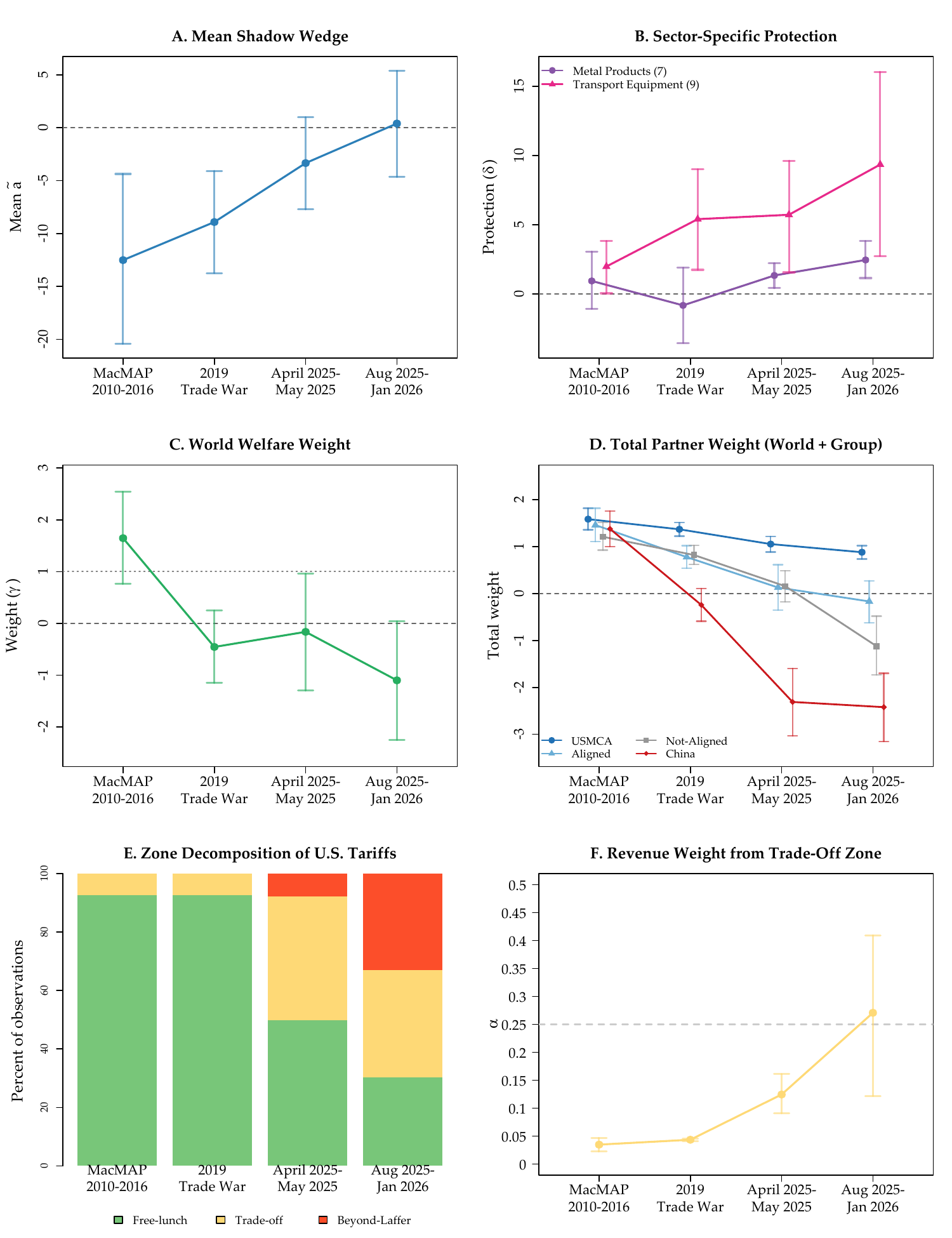}
\end{figure}

Panel~A reports the mean marginal shadow wedge ($\bar{\omega}$). This mean wedge shifts steadily toward aggression over time: tariffs move from being far below the domestically efficient levels in the pre-2019 baseline toward levels more consistent with, or exceeding, domestic welfare-revenue motives by late 2025. Panel~B isolates sector-specific protection wedges ($\kappa^s$) for strategic industries such as Metal Products and Transport Equipment, showing marked increases in protectionist motives beyond standard economic considerations.

Panel~C tracks the revealed baseline weight on foreign welfare ($\gamma^W$), which exceeded unity prior to 2019 (indicating the U.S.\ fully internalized the impact of tariffs on foreign welfare) but declines sharply during the 2025 trade war, falling below zero by late 2025. Panel~D decomposes the total effective weight on targeted-partner welfare ($\gamma^W + \psi^{g(j)}$), highlighting that while USMCA and aligned partners broadly track the global baseline, China receives a sharply negative group-specific weight. At its most negative, the U.S.\ was willing to sacrifice its own domestic objectives to reduce Chinese welfare, revealing a strong preference for punitive tariffs.

Panel~E connects these shifts back to our MFEI framework, decomposing observed tariffs into the three fiscal zones. It illustrates the escalating severity of the trade policy: while nearly all pre-2019 tariffs were in the Free-Lunch zone, the share of partner-sector tariffs residing in the Beyond-Laffer zone grew to more than 20 percent by late 2025. Finally, Panel~F presents the diagnostic estimates of the implied revenue weight ($\tilde{\alpha}_{r(d)}$) for tariffs residing strictly in the Trade-Off zone.\footnote{We exclude China as an outlier for 2019 Trade War tariffs, see Supplemental Appendix \ref{app-sec:inv_opt_timelines}.} Prior to 2025, the revealed weight on revenue was effectively zero. By late 2025, these estimates rise significantly, converging toward the OMB benchmark of 0.25, indicating an increased role for revenue motives in late-2025 tariff-setting.

These estimates should be interpreted as reduced-form wedges. A negative weight on a targeted partner's welfare may reflect punitive intent, but it may also absorb misspecification, adjustment costs, or omitted political-economy channels. Nonetheless, the time-series pattern is clear: by late 2025, the observed U.S.\ tariff schedule is difficult to rationalize without (i) diminished concern for foreign welfare, (ii) differential treatment of China, and (iii) a significantly higher marginal value placed on tariff revenue.

\section{Concluding Remarks}
\label{sec:conclusion}

This paper quantifies the fiscal and welfare limits of protectionism by computing the tariff Laffer curve for the United States during the 2025 trade war. We find that tariff revenues peak at rates between 20 and 30 percent, while welfare peaks between 0 and 10 percent. The gap between these two peaks defines the Trade-Off zone where the United States can raise revenue only at the cost of welfare. 

A central finding is the asymmetry between the revenue and welfare effects of retaliation. U.S.\ tariff revenue increases monotonically with the number of trading partners targeted and is only modestly reduced by retaliation. U.S.\ welfare, in contrast, is sharply sensitive to retaliation: if all trading partners match U.S.\ tariff rate increases, U.S.\ welfare gains vanish entirely and turn into large welfare losses above single-digit tariff rates. 

Our inverse-optimum analysis reveals a marked shift in U.S.\ trade policy objectives. In our data, many 2025 tariffs lie in the welfare-revenue Trade-Off zone, and China's tariffs lie well beyond its Laffer peak. Prior to 2025, tariffs were consistent with the government placing positive weight on foreign welfare and near-zero independent weight on tariff revenue. By late 2025, the baseline concern for foreign welfare had diminished, the revealed weight on tariff revenue had risen toward public-finance benchmarks, and the total weight on Chinese welfare had turned sharply negative. The polarization of total partner weights suggests a shift in the revealed objectives of U.S.\ trade policy that is consistent with tariff policy being wielded as an instrument of geoeconomic power \citep{ClaytonMaggioriSchreger2025macro,ClaytonMaggioriSchreger2026}.

Several extensions would strengthen these findings. Our model is static and abstracts from the dynamic adjustment of trade flows and reallocation \citep{NBERw33792, alessandriamixfriends}. The model also treats tariffs as uniform within each sector-partner pair, whereas the actual tariff schedule exhibits substantial product-level variation. Finally, decomposing our reduced-form shadow wedges into structural and strategic channels may provide a deeper account of the forces that shifted U.S.\ trade policy in 2025.

\clearpage

\linespread{1.0}
\bibliographystyle{style1}
\bibliography{bib_master}

\newpage
\thispagestyle{empty}
\vspace*{\fill}
\begin{center}
\Large\textbf{Supplemental Appendix:}
\end{center}
\vspace*{\fill}
\clearpage

\newpage
\appendix

\setcounter{table}{0}
\setcounter{figure}{0}
\renewcommand{\thetable}{\Alph{section}.\arabic{table}}
\renewcommand{\thefigure}{\Alph{section}.\arabic{figure}}
\renewcommand{\thesection}{\Alph{section}}
\renewcommand{\thesubsection}{\Alph{section}.\arabic{subsection}}

\newpage

\section{Marginal Excess Burden, Marginal Cost of Public Funds, and the MFEI}
\label{appendix:MCF_connection}

In this appendix, we formally connect our tariff-revenue metric, the Marginal Fiscal Efficiency Index (MFEI), to the standard public finance concepts of the marginal excess burden (MEB) and the marginal cost of public funds (MCF), following the taxonomy in \citet{Dahlby2008}. To theoretically ground our quantitative results, we derive explicit analytical expressions that link the boundaries of fiscal efficiency directly to the structural parameters of our general equilibrium model, specifically the trade elasticity. Finally, we formally establish the properties of the MFEI and show how it resolves the asymptotic discontinuity of the MEB at the Laffer peak.

\subsection{Definitions and Accounting Identities}

Consider a government that imposes an ad-valorem tariff $t$, which we express as a gross factor $\tau = 1 + t$. Let $V(\tau)$ denote aggregate private welfare (real income before any lump-sum rebate of tariff revenue) and $R(\tau)$ denote government tariff revenue. The Marginal Cost of Public Funds (MCF) is defined as the gross private welfare cost per additional dollar of revenue raised:
\begin{equation}
\text{MCF}(\tau) \equiv \frac{-dV/d\tau}{dR/d\tau}.
\label{eq:MCF_appendix}
\end{equation}
Of each dollar raised, one dollar is transferred mechanically from the private sector to the government. The Marginal Excess Burden (MEB) isolates the deadweight cost \emph{in excess} of this transfer:
\begin{equation}
\text{MEB}(\tau) \equiv \text{MCF}(\tau) - 1 = \frac{-(dV/d\tau + dR/d\tau)}{dR/d\tau}.
\label{eq:MEB_MCF_relation}
\end{equation}

In our general equilibrium model, total domestic U.S.\ welfare $W(\tau)$ (measured as real income) includes the lump-sum rebate of tariff revenue to households, such that the accounting identity $W(\tau) = V(\tau) + R(\tau)$ holds. Differentiating with respect to $\tau$ and substituting into the MEB definition yields a direct identity linking the MEB to the marginal change in total aggregate welfare:
\begin{equation}
\text{MEB}(\tau) = \frac{-dW/d\tau}{dR/d\tau}.
\label{eq:MEB_W}
\end{equation}

In standard domestic taxation, $-dW/d\tau$ strictly captures the deadweight loss (Harberger triangles) from the marginal tax increase, implying that $\text{MEB} > 0$. However, in an international setting, $-dW/d\tau$ captures both the domestic deadweight loss \emph{and} terms-of-trade transfers between countries. A tariff that depresses the world price of imports transfers welfare from foreign exporters to domestic consumers. At low tariff rates, this terms-of-trade gain can exceed the domestic deadweight loss, generating a positive marginal effect of the tariff on total welfare ($dW/d\tau > 0$). When revenue is also increasing ($dR/d\tau > 0$), Equation (\ref{eq:MEB_W}) implies that $\text{MEB}(\tau) < 0$. This negative excess burden is the tariff-policy equivalent of the classic optimal tariff argument \citep{johnsonOptimumTariffsRetaliation1953} and explains why our MEB curves take negative values in the ``win-win'', or ``Free-Lunch'' zone, a feature absent from standard domestic tax analysis.

\subsection{The Marginal Fiscal Efficiency Index (MFEI)}
\label{app-sec:MFEI}
Let $\alpha \ge 0$ denote the MEB of the existing domestic tax system---the welfare cost of the marginal dollar raised through broad-based distortionary instruments like income or payroll taxes. A tariff is fiscally neutral when its MEB aligns with the broader tax system, i.e., when
\begin{equation}
\text{MEB}_{\text{tariff}}(\tau^{FE}) = \alpha.
\label{eq:efficient_tariff}
\end{equation}
We calibrate $\alpha = 0.25$ following OMB Circular A-94 \citep{omb2023a94}. 

The government's unconstrained fiscal objective is formalized by the Lagrangian $\mathcal{L}(\tau) = W(\tau) + \alpha \cdot R(\tau)$. The first-order condition $d\mathcal{L}/d\tau = 0$ is exactly equivalent to the equalization of the MCF across instruments, $\text{MEB}_{\text{tariff}} = \alpha$. To verify this, observe that:
\begin{equation}
\label{eq:MEB_fiscal_eff}
\frac{dW}{d\tau} + \alpha \frac{dR}{d\tau} = 0 \quad \Longleftrightarrow \quad -\frac{dW/d\tau}{dR/d\tau} = \alpha \quad \Longleftrightarrow \quad \text{MEB}(\tau) = \alpha.
\end{equation}

A practical mathematical difficulty with plotting and analyzing the standard MEB across a wide grid of tariffs is that it exhibits an asymptotic discontinuity exactly at the Laffer peak (where $dR/d\tau \to 0$, causing $\text{MEB} \to \pm\infty$). To continuously map the entire global tariff schedule into a bounded metric, we construct the MFEI by normalizing the first-order condition using the $L^1$ norm of its components:
\begin{equation}
\text{MFEI}(\tau) \equiv \frac{dW/d\tau + \alpha \cdot dR/d\tau}{|dW/d\tau| + \alpha \cdot |dR/d\tau|}.
\end{equation}
We handle the edge case where welfare and revenue have the same optimum by defining $\text{MFEI}(\tau) = 0$ when the denominator is zero, i.e.,
\begin{equation}
\text{MFEI}(\tau) \equiv 0, \qquad \text{if} \qquad |dW/d\tau| + \alpha \cdot |dR/d\tau| = 0,
\end{equation}
indicating that the marginal effect of the tariff is completely fiscally neutral.

\begin{proposition}[Fiscal Regions of the MFEI]\label{prop:MFEI_Zones}
The MFEI maps the unbounded marginal excess burden into a bounded index $\text{MFEI}(\tau) \in [-1, 1]$ that cleanly partitions the tariff schedule into five distinct fiscal regions:
\begin{enumerate}
    \item \textbf{\text{Free-Lunch Zone,} $\text{MFEI} = 1$:} Occurs if and only if $dW/d\tau \ge 0$ and $dR/d\tau > 0$. Both aggregate welfare and tariff revenue are non-decreasing.
    \item \textbf{\text{Fiscally Efficient Trade-Off,} $\text{MFEI} \in (0, 1)$:} Occurs if and only if $dW/d\tau < 0$, $dR/d\tau > 0$, and $\text{MEB}(\tau) < \alpha$. The tariff imposes a private welfare cost, but its cost per dollar of revenue raised is strictly lower than that of the existing domestic tax system. The government is strictly better off raising marginal funds via the tariff.
    \item \textbf{\text{Fiscally Neutral Tariff} $\text{MFEI} = 0$:} Occurs if and only if (i) $dW/d\tau < 0$, $dR/d\tau > 0$, and $\text{MEB}(\tau) = \alpha$, i.e., the marginal welfare cost per dollar of tariff revenue perfectly matches the marginal excess burden of the domestic tax system; or (ii) $|dW/d\tau| + \alpha \cdot |dR/d\tau| = 0$.
    \item \textbf{\text{Fiscally Inefficient Trade-Off,} $\text{MFEI} \in (-1, 0)$:} Occurs if and only if $dW/d\tau < 0$, $dR/d\tau > 0$, and $\text{MEB}(\tau) > \alpha$. The tariff generates additional revenue, but its marginal welfare cost strictly exceeds that of the existing tax system. The government would be strictly better off raising these funds through domestic taxation.
    \item \textbf{\text{Beyond-Laffer Zone,} $\text{MFEI} = -1$:} Occurs if and only if $dW/d\tau < 0$ and $dR/d\tau \le 0$. The tariff strictly exceeds the revenue-maximizing rate; further tariff increases destroy both welfare and revenue.
\end{enumerate}
\end{proposition}

\begin{proof}
Because $\alpha \ge 0$, the denominator is strictly positive as long as either welfare or revenue is responsive to the tariff. By the triangle inequality, the index is strictly bounded in $[-1, 1]$. We proceed by evaluating the index algebraically across the regions:
\begin{itemize}
    \item \textbf{Zone 1:} If $dW/d\tau \ge 0$ and $dR/d\tau > 0$, both terms in the numerator are non-negative, meaning the absolute values in the denominator evaluate to the exact base values in the numerator. The ratio evaluates to exactly $1$.
    \item \textbf{Zones 2, 3, and 4 (The Trade-Off Regions):} Suppose $dW/d\tau < 0$ and $dR/d\tau > 0$. The absolute values evaluate to $|dW/d\tau| = -dW/d\tau$ and $|dR/d\tau| = dR/d\tau$. Dividing the numerator and denominator by $dR/d\tau > 0$ maps the relationship directly to the MEB:
    \begin{equation*}
    \text{MFEI}(\tau) = \frac{(dW/d\tau)/(dR/d\tau) + \alpha}{-(dW/d\tau)/(dR/d\tau) + \alpha} = \frac{\alpha - \text{MEB}(\tau)}{\alpha + \text{MEB}(\tau)}.
    \end{equation*}
    Because $dW/d\tau < 0$ and $dR/d\tau > 0$, we know $\text{MEB}(\tau) > 0$. Since $\alpha > 0$, the denominator $\alpha + \text{MEB}(\tau)$ is strictly positive. The sign and magnitude of the MFEI therefore depend entirely on the numerator:
    \begin{itemize}
        \item If $\text{MEB}(\tau) < \alpha$, then $0 < \alpha - \text{MEB}(\tau) < \alpha + \text{MEB}(\tau)$. Thus, $\text{MFEI}(\tau) \in (0, 1)$, establishing Zone 2.
        \item If $\text{MEB}(\tau) = \alpha$, then $\alpha - \text{MEB}(\tau) = 0$. Thus, $\text{MFEI}(\tau) = 0$. This perfectly recovers the first-order condition in equation \eqref{eq:MEB_fiscal_eff}, establishing Zone 3.  In the edge case where $|dW/d\tau| + \alpha \cdot |dR/d\tau| = 0$, then $\text{MFEI}(\tau) = 0$ by definition.
        \item If $\text{MEB}(\tau) > \alpha$, then $\alpha - \text{MEB}(\tau) < 0$. The ratio is negative, but since $|\alpha - \text{MEB}(\tau)| = \text{MEB}(\tau) - \alpha < \text{MEB}(\tau) + \alpha$, it is strictly greater than $-1$. Thus, $\text{MFEI}(\tau) \in (-1, 0)$, establishing Zone 4.
    \end{itemize}
    \item \textbf{Zone 5:} If $dW/d\tau < 0$ and $dR/d\tau \le 0$, both terms in the numerator are non-positive (and at least one is strictly negative). The absolute values in the denominator evaluate to the exact negative of the numerator. The ratio evaluates to exactly $-1$. Note that this limit is well-defined and continuous across the Laffer peak where $dR/d\tau \to 0$. \qedhere
\end{itemize}
\end{proof}

\begin{remark*}[Extension to Externality-Correcting Taxes and Subsidies]
The MFEI normalizes the government's fiscal objective by the $L^1$ norm of its components; therefore, the index symmetrically accommodates policies that expend revenue to generate welfare ($dW/dx > 0$ and $dR/dx < 0$), such as Pigouvian subsidies or externality-correcting taxes operating beyond their own Laffer peaks. 

No additional definition is required:
\begin{equation*}
\text{MFEI}(x) = \frac{dW/dx + \alpha \cdot dR/dx}{dW/dx - \alpha \cdot dR/dx} = \frac{\frac{dW/dx}{-dR/dx} - \alpha}{\frac{dW/dx}{-dR/dx} + \alpha},
\end{equation*}
where the second equality follows from dividing the numerator and 
denominator by $-dR/dx > 0$. This immediately mirrors the Trade-Off zone expressions in Proposition~\ref{prop:MFEI_Zones}.  Therefore, a revenue-decreasing policy is fiscally efficient $\left(\text{MFEI} \in (0,1)\right)$ if and only if its marginal welfare benefit per dollar of forgone revenue strictly exceeds the marginal excess burden of the domestic tax system required to replace those funds $\left(\frac{dW/dx}{-dR/dx} > \alpha\right)$.
\end{remark*}

\begin{proposition}[Limiting Behavior as the Revenue Weight Varies]\label{prop:MFEI_limits}
As the marginal excess burden of the domestic tax system ($\alpha$) varies, the MFEI smoothly encompasses the polar objectives of pure welfare maximization and pure revenue maximization:
\begin{enumerate}
    \item \textbf{Pure Welfare Maximization ($\alpha \to 0$):} $\lim_{\alpha \to 0} \text{MFEI}(\tau) = \text{sgn}\left(dW/d\tau\right)$ (for $dW/d\tau \neq 0$). The index evaluates the tariff solely based on its marginal welfare effect, placing zero weight on revenue. The Fiscally Efficient Tariff ($\text{MFEI}=0$) converges to the classical optimal tariff.
    \item \textbf{Pure Revenue Maximization ($\alpha \to \infty$):} $\lim_{\alpha \to \infty} \text{MFEI}(\tau) = \text{sgn}\left(dR/d\tau\right)$ (for $dR/d\tau \neq 0$). The index evaluates the tariff solely based on its marginal revenue effect, placing zero weight on private welfare costs. The Fiscally Efficient Tariff ($\text{MFEI}=0$) converges to the Laffer peak.
\end{enumerate}
\end{proposition}

\begin{proof}
For $\alpha \to 0$: Assuming the tariff has a non-zero marginal effect on welfare ($dW/d\tau \neq 0$), taking the limit as $\alpha \to 0$ eliminates the revenue terms from both the numerator and denominator:
\begin{equation*}
\lim_{\alpha \to 0} \text{MFEI}(\tau) = \frac{dW/d\tau}{|dW/d\tau|} = \text{sgn}\left(\frac{dW}{d\tau}\right).
\end{equation*}
This evaluates exactly to the sign function of marginal welfare. It evaluates to $1$ when welfare is increasing and $-1$ when decreasing, jumping precisely at the welfare peak.

For $\alpha \to \infty$: Assuming marginal revenue is non-zero ($dR/d\tau \neq 0$), we factor out $\alpha$ by dividing the numerator and denominator by it:
\begin{equation*}
\lim_{\alpha \to \infty} \text{MFEI}(\tau) = \lim_{\alpha \to \infty} \frac{\frac{1}{\alpha} \frac{dW}{d\tau} + \frac{dR}{d\tau}}{\frac{1}{\alpha} \left| \frac{dW}{d\tau} \right| + \left| \frac{dR}{d\tau} \right|} = \frac{dR/d\tau}{|dR/d\tau|} = \text{sgn}\left(\frac{dR}{d\tau}\right).
\end{equation*}
This evaluates exactly to the sign function of marginal revenue: $1$ on the upward-sloping side of the Laffer curve and $-1$ on the downward-sloping side, transitioning sharply at the revenue-maximizing tariff. \qedhere
\end{proof}

\begin{proposition}[Monotonicity Properties]\label{prop:MFEI_monotonicity}
Fix $\alpha>0$. Within the trade-off zones ($dW/d\tau < 0$ and $dR/d\tau > 0$), the MFEI satisfies the following monotonicity properties:
\begin{enumerate}
    \item $\partial \text{MFEI}/\partial \text{MEB} < 0$ (strictly decreasing in the tariff's marginal excess burden).
    \item $\partial \text{MFEI}/\partial \alpha > 0$ (strictly increasing in the revenue-weight).
\end{enumerate}
\end{proposition}

\begin{proof}
By Proposition \ref{prop:MFEI_Zones}, the MFEI in the Trade-Off zone is given by $\text{MFEI}(\tau) = \frac{\alpha - \text{MEB}(\tau)}{\alpha + \text{MEB}(\tau)}$. 
Taking the partial derivative with respect to $\text{MEB}$ yields:
\begin{equation*}
\frac{\partial \text{MFEI}}{\partial \text{MEB}} = \frac{-1 \cdot (\alpha + \text{MEB}) - 1 \cdot (\alpha - \text{MEB})}{(\alpha + \text{MEB})^2} = \frac{-2\alpha}{(\alpha + \text{MEB})^2}.
\end{equation*}
Since $\alpha > 0$, this derivative is strictly negative, confirming that higher marginal excess burdens map to monotonically lower fiscal efficiency indices.

Taking the partial derivative with respect to $\alpha$ yields:
\begin{equation*}
\frac{\partial \text{MFEI}}{\partial \alpha} = \frac{1 \cdot (\alpha + \text{MEB}) - 1 \cdot (\alpha - \text{MEB})}{(\alpha + \text{MEB})^2} = \frac{2\text{MEB}}{(\alpha + \text{MEB})^2}.
\end{equation*}
Since $\text{MEB} > 0$ in the Trade-Off zone, this derivative is strictly positive. A higher relative cost of domestic public funds mechanically makes the tariff fiscally efficient. Note that outside of the Trade-Off zone, $\text{MFEI}$ is constant. \qedhere
\end{proof}

\subsection{Analytical Connection to the Trade Elasticity}

To theoretically characterize how the MEB behaves across the tariff schedule, we evaluate it analytically as a function of the underlying model elasticities. %

\begin{proposition}[Welfare- and Revenue-Maximizing Tariffs]\label{prop:MEB_analytical}
Assume a domestic import demand elasticity $\varepsilon > 1$ and a foreign export supply elasticity $\varsigma > 0$. The Marginal Excess Burden of a tariff factor $\tau \ge 1$ is given by:
\begin{equation}
\text{MEB}(\tau) = \frac{ \varepsilon \big[ (\tau - 1) \varsigma - 1 \big] }{ \varepsilon + \tau \varsigma - (\tau - 1) \varepsilon \varsigma }.
\label{eq:MEB_analytic}
\end{equation}
Furthermore, the domestic welfare-maximizing tariff ($\tau^*_\text{Welfare}$) and the revenue-maximizing tariff ($\tau^*_{\text{Laffer}}$) are respectively given by:
\begin{equation}
\tau^*_\text{Welfare} - 1 = \frac{1}{\varsigma}, \quad \quad \tau^*_{\text{Laffer}} - 1 = \frac{\varsigma + \varepsilon}{\varsigma(\varepsilon - 1)}.
\label{eq:optimal_laffer}
\end{equation}
\end{proposition}

\begin{proof}
Consider the welfare decomposition of a marginal tariff change. Let $p^F$ denote the world export price, so the domestic purchaser price is $p = \tau p^F$, and let $M$ denote the import volume. The marginal change in real welfare is given by the sum of the terms-of-trade effect and the volume-of-trade effect (deadweight loss):
\begin{equation}
\frac{dW}{d\tau} = - M \frac{dp^F}{d\tau} + (\tau - 1) p^F \frac{dM}{d\tau}.
\label{eq:dW_decomp}
\end{equation}
The differential change in revenue, $R = (\tau - 1) p^F M$, is:
\begin{equation}
\frac{dR}{d\tau} = p^F M + (\tau - 1) M \frac{dp^F}{d\tau} + (\tau - 1) p^F \frac{dM}{d\tau}.
\label{eq:dR_decomp}
\end{equation}

By definition, $p = \tau p^F$, which implies the logarithmic differential $d\ln p = d\tau/\tau + d\ln p^F$. Market clearing requires that the partial change in import demand ($-\varepsilon \equiv \partial \ln M / \partial \ln p$) equals the change in export supply ($\varsigma \equiv \partial \ln M / \partial \ln p^F$):
\begin{equation}
d\ln M = -\varepsilon \left(\frac{d\tau}{\tau} + d\ln p^F\right) = \varsigma d\ln p^F.
\end{equation}
Solving for the equilibrium price and quantity responses yields:
\begin{equation}
\frac{d\ln p^F}{d\tau} = -\frac{\varepsilon}{\tau(\varsigma + \varepsilon)}, \quad \quad \frac{d\ln M}{d\tau} = -\frac{\varepsilon \varsigma}{\tau(\varsigma + \varepsilon)}.
\label{eq:dp_dM}
\end{equation}
Using $dp^F/d\tau = p^F (d\ln p^F/d\tau)$ and $dM/d\tau = M (d\ln M/d\tau)$, we substitute (\ref{eq:dp_dM}) into the marginal welfare decomposition (\ref{eq:dW_decomp}):
\begin{equation}
\frac{dW}{d\tau} = \frac{p^F M}{\tau (\varsigma + \varepsilon)} \Big[ \varepsilon - (\tau - 1) \varepsilon \varsigma \Big].
\label{eq:dW_dtau_analytic}
\end{equation}
Similarly, substituting into the marginal revenue decomposition (\ref{eq:dR_decomp}) gives:
\begin{equation}
\frac{dR}{d\tau} = \frac{p^F M}{\tau (\varsigma + \varepsilon)} \Big[ \tau (\varsigma + \varepsilon) - (\tau - 1) \varepsilon (1 + \varsigma) \Big] = \frac{p^F M}{\tau (\varsigma + \varepsilon)} \Big[ \varepsilon + \tau \varsigma - (\tau - 1) \varepsilon \varsigma \Big].
\label{eq:dR_dtau_analytic}
\end{equation}
Taking the negative ratio $\text{MEB}(\tau) = -(dW/d\tau)/(dR/d\tau)$ exactly yields Equation (\ref{eq:MEB_analytic}).

The welfare-maximizing optimal tariff occurs where $dW/d\tau = 0$ (and thus $\text{MEB}(\tau) = 0$). Setting the numerator of (\ref{eq:MEB_analytic}) to zero recovers the classic inverse-elasticity rule $\tau^*_\text{Welfare} - 1 = 1/\varsigma$. The revenue-maximizing tariff (the Laffer peak) occurs where $dR/d\tau = 0$. Setting the denominator of (\ref{eq:MEB_analytic}) to zero and rearranging for $\tau-1$ yields the expression for $\tau^*_{\text{Laffer}}$ in (\ref{eq:optimal_laffer}). \qedhere
\end{proof}

Proposition \ref{prop:MEB_analytical} provides a structural interpretation for the shape of the tariff Laffer curve. In the context of a single-sector Ricardian model, the partial elasticity of import demand is roughly proportional to the trade elasticity $\theta$, taking the form $\varepsilon \approx \theta (1 - \pi) + 1$, where $\pi$ is the domestic expenditure share \citep{CALIENDO2024103820}. 

Equation (\ref{eq:optimal_laffer}) mathematically formalizes the numerical robustness exercises in Supplemental Appendix \ref{app:robustness_elasticities}. As the trade elasticity $\theta$ increases, products become more substitutable, which mechanically increases both $\varepsilon$ and $\varsigma$. A higher $\varsigma$ diminishes the terms-of-trade advantage, shifting the optimal tariff $\tau^*_\text{Welfare}$ inward toward free trade. Simultaneously, a higher $\varepsilon$ amplifies the deadweight loss and accelerates the erosion of the tax base, which strictly decreases $\tau^*_{\text{Laffer}}$. Consequently, a higher $\theta$ systematically compresses the entire Laffer curve.

\begin{proposition}[Strict Monotonicity and the Fiscally Optimal Tariff]\label{prop:MFEI_tau_monotonicity}
Maintain the constant-elasticity environment of Proposition \ref{prop:MEB_analytical} with $\varepsilon > 1$, $\varsigma > 0$, and fix $\alpha > 0$. Then:
\begin{enumerate}
    \item \textbf{Strict Monotonicity in $\tau$:} The marginal excess burden is strictly increasing in $\tau$ away from the Laffer peak. For all $\tau \neq \tau^*_{\text{Laffer}}$:
    \begin{equation}
    \frac{d\,\text{MEB}(\tau)}{d\tau} = \frac{\varepsilon\,\varsigma\,(1+\varsigma)}{\big(\varepsilon+\tau\varsigma-(\tau-1)\varepsilon\varsigma\big)^2} > 0.
    \end{equation}
    Consequently, on the interior trade-off interval $(\tau^*_\text{Welfare}, \tau^*_{\text{Laffer}})$, $\text{MFEI}(\tau)$ is strictly monotonically decreasing in $\tau$.
    \item \textbf{Closed-Form Fiscally Optimal Tariff:} There exists a unique fiscally optimal tariff $\tau^{FE} \in (\tau^*_\text{Welfare}, \tau^*_{\text{Laffer}})$ satisfying $\text{MEB}(\tau^{FE}) = \alpha$ (equivalently, $\text{MFEI}(\tau^{FE}) = 0$), given by the closed-form analytical expression:
    \begin{equation}
    \tau^{FE} = \frac{\varepsilon(1+\varsigma)(1+\alpha)}{\varsigma\big(\varepsilon+\alpha(\varepsilon-1)\big)}.
    \end{equation}
\end{enumerate}
\end{proposition}

\begin{proof}
For (1), write $\text{MEB}(\tau) = N(\tau)/D(\tau)$ using Equation (\ref{eq:MEB_analytic}), where
\begin{equation*}
N(\tau) = \varepsilon\big((\tau-1)\varsigma-1\big),
\qquad
D(\tau) = \varepsilon+\tau\varsigma-(\tau-1)\varepsilon\varsigma.
\end{equation*}
Applying the quotient rule,
\begin{equation*}
\frac{d}{d\tau}\left(\frac{N(\tau)}{D(\tau)}\right)
=
\frac{N'(\tau)D(\tau)-N(\tau)D'(\tau)}{D(\tau)^2}.
\end{equation*}
Compute the derivatives:
\begin{equation*}
N'(\tau)=\varepsilon\varsigma,
\qquad
D'(\tau)=\varsigma-\varepsilon\varsigma = -(\varepsilon-1)\varsigma.
\end{equation*}
Substituting,
\begin{align*}
\frac{d\,\text{MEB}(\tau)}{d\tau}
&=
\frac{\varepsilon\varsigma\,D(\tau)-\varepsilon\big((\tau-1)\varsigma-1\big)\big(-(\varepsilon-1)\varsigma\big)}{D(\tau)^2} \\
&=
\frac{\varepsilon\varsigma\left[D(\tau)+(\varepsilon-1)\big((\tau-1)\varsigma-1\big)\right]}{D(\tau)^2}.
\end{align*}
Finally, simplify the bracketed term:
\begin{align*}
D(\tau)+(\varepsilon-1)\big((\tau-1)\varsigma-1\big)
&=
\left[\varepsilon+\tau\varsigma-(\tau-1)\varepsilon\varsigma\right]
+(\varepsilon-1)(\tau-1)\varsigma-(\varepsilon-1) \\
&=
1+\tau\varsigma-(\tau-1)\varsigma
 \\
&=1+\varsigma.
\end{align*}
Therefore,
\begin{equation*}
\frac{d\,\text{MEB}(\tau)}{d\tau}
=
\frac{\varepsilon\varsigma(1+\varsigma)}{D(\tau)^2}
=
\frac{\varepsilon\,\varsigma\,(1+\varsigma)}{\big(\varepsilon+\tau\varsigma-(\tau-1)\varepsilon\varsigma\big)^2}
>0,
\end{equation*}
for all $\tau\neq\tau^*_{\text{Laffer}}$ where $D(\tau)\neq 0$.

Within the trade-off interval $(\tau^*_\text{Welfare}, \tau^*_{\text{Laffer}})$, we have $\text{MEB}(\tau) > 0$. By Proposition \ref{prop:MFEI_monotonicity}, $\partial \text{MFEI}/\partial \text{MEB} < 0$. By the chain rule,
$d\text{MFEI}/d\tau = (\partial \text{MFEI}/\partial \text{MEB}) \cdot (d\text{MEB}/d\tau) < 0$.

For (2), existence and uniqueness of $\tau^{FE}$ follow because $\text{MEB}(\tau)$ is continuous and strictly increasing on $(\tau^*_\text{Welfare}, \tau^*_{\text{Laffer}})$, mapping from $0$ (at $\tau^*_\text{Welfare}$) to $+\infty$ (as $\tau \to \tau^*_{\text{Laffer}}$). Setting Equation (\ref{eq:MEB_analytic}) equal to $\alpha$ yields:
\begin{equation*}
\varepsilon \big[ (\tau - 1) \varsigma - 1 \big] = \alpha \big[ \varepsilon + \tau \varsigma - (\tau - 1) \varepsilon \varsigma \big].
\end{equation*}
Substituting $\tau = (\tau - 1) + 1$ on the right side and collecting all terms containing $(\tau - 1)$ yields:
\begin{equation*}
(\tau - 1) \varsigma \big[ \varepsilon - \alpha(1 - \varepsilon) \big] = \varepsilon(1 + \alpha) + \alpha \varsigma.
\end{equation*}
Noting that $\varepsilon - \alpha(1 - \varepsilon) = \varepsilon + \alpha(\varepsilon - 1)$ and isolating $(\tau - 1)$ yields:
\begin{equation*}
\tau - 1 = \frac{\varepsilon(1+\alpha) + \alpha\varsigma}{\varsigma\big(\varepsilon+\alpha(\varepsilon-1)\big)}.
\end{equation*}
Adding $1 = \frac{\varsigma(\varepsilon+\alpha(\varepsilon-1))}{\varsigma(\varepsilon+\alpha(\varepsilon-1))}$ directly yields
$\tau^{FE} = \frac{\varepsilon(1+\varsigma)(1+\alpha)}{\varsigma(\varepsilon+\alpha(\varepsilon-1))}$. \qedhere
\end{proof}

\section{Robustness of Laffer Curve Peaks to Elasticities and Input--Output Linkages}
\label{app:robustness_elasticities}

Our results connect to a large quantitative literature that uses multi-country general-equilibrium models to compute Nash and cooperative tariffs and to quantify the welfare costs of trade wars. A central benchmark is \citet{Ossa2014}, who finds that optimal tariffs can be upwards of 60 percent in calibrated trade-war environments.  Methodologically, much of the optimal-tariff literature characterizes a single optimum schedule using first-order conditions, sometimes in sufficient-statistic form (\citealt{acr2012}; \citealt{costinotChapterTradeTheory2014}; \citealt{costinot2015comparative}). In contrast, we compute the full relationship between tariffs, revenue, and welfare over the entire policy-relevant range for the entire world. We do so by repeatedly solving the full general equilibrium model on a dense grid of U.S.\ tariffs and potential partner responses, providing an exhaustive grid search that involves solving the general equilibrium model millions of times in parallel.

Our baseline calibration uses sector-specific trade elasticities from \cite{Chen2023Fragmentation}, which range from approximately 2 to 15 across traded goods sectors, combined with the full input--output structure from the Eora tables. In Section~\ref{sec:introduction}, we noted that our welfare-maximizing tariff rates (0--10 percent) are substantially lower than the optimal tariffs found in earlier quantitative trade-war models, most notably \citet{Ossa2014}, who reports Nash optimal tariffs averaging approximately 60 percent.  We conduct several quantitative experiments to disentangle the source of this discrepancy.  

One natural hypothesis is that the lower optimal tariffs we find could be driven by differences in trade elasticities. We investigate the effects of trade elasticities by re-computing the bilateral Laffer and welfare curves under two alternative specifications.  Figure \ref{fig:Laffer_Sigma4} shows the equivalent of Figure \ref{fig:usa_laffer_eu_chn} when we calibrate the model using a uniform trade elasticity of $\theta^s = 4$ for all sectors, chosen to match the mean trade elasticity in \citet{SIMONOVSKA201434}. Figure \ref{fig:Laffer_Sigma24} shows the same results with a uniform trade elasticity of $\theta^s = 2.42$ for all sectors, chosen to match the mean trade elasticity in \citet{Ossa2014}.  These figures show that tariff revenue is highly responsive to the trade elasticity shifting the Laffer peak tariff rate markedly higher.  Conversely, an elasticity of $\theta^s = 4$ has a minimal impact on the welfare-maximizing tariff rate, whereas an elasticity $\theta^s = 2.42$ raises the optimal tariff rate from 9--10 percent up to 19--20 percent, a significant increase, but far lower than the 60+ percent \citet{Ossa2014} finds.

\begin{figure}[htbp]
    \centering
    \caption{U.S.\ Laffer Curve with E.U. and China ($\theta^s=4$)}
    \label{fig:Laffer_Sigma4}
    \includegraphics[width=\textwidth]{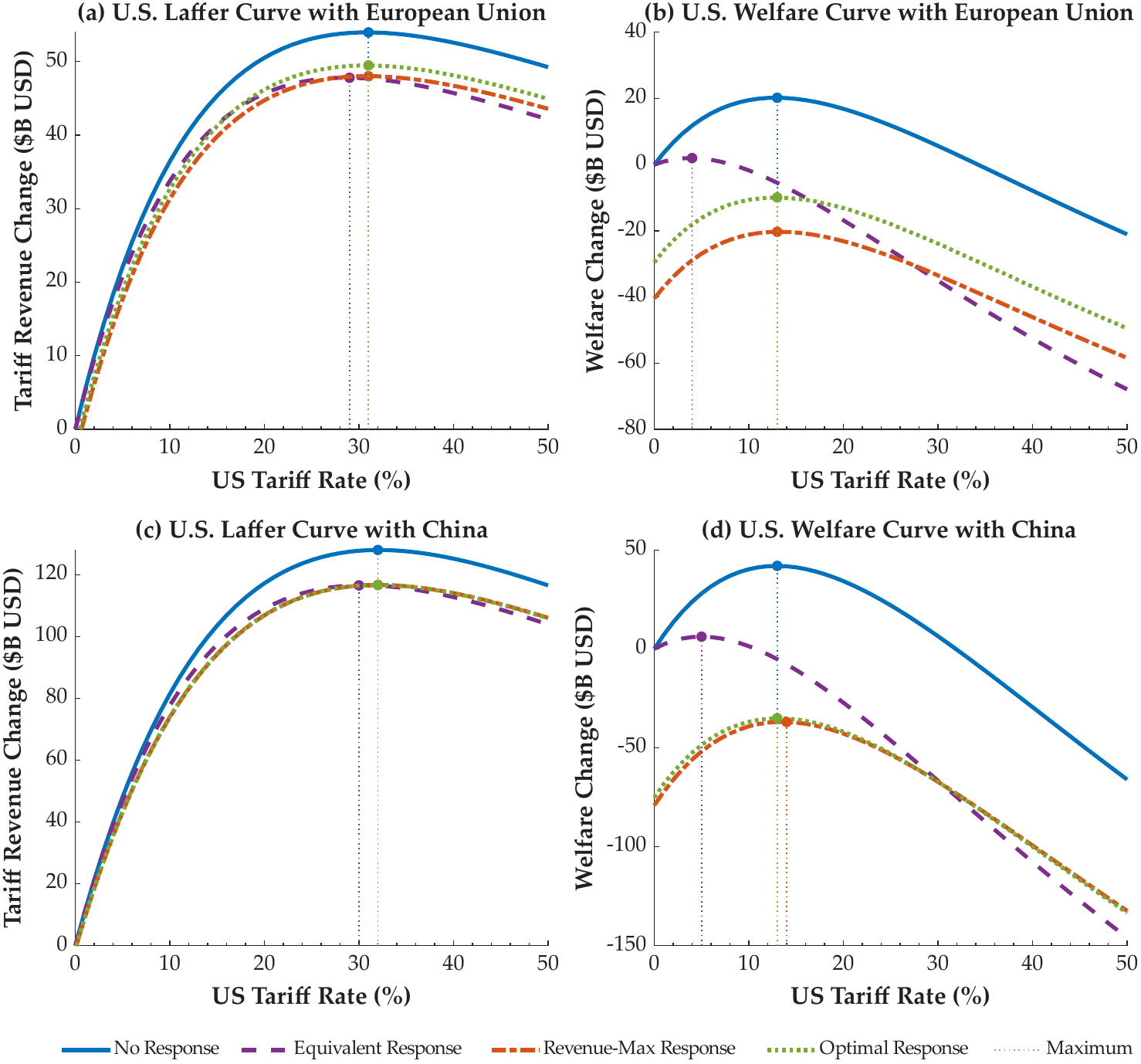}
\end{figure}

\begin{figure}[htbp]
    \centering
    \caption{U.S.\ Laffer Curve with E.U. and China ($\theta^s=2.42$)}
    \label{fig:Laffer_Sigma24}
    \includegraphics[width=\textwidth]{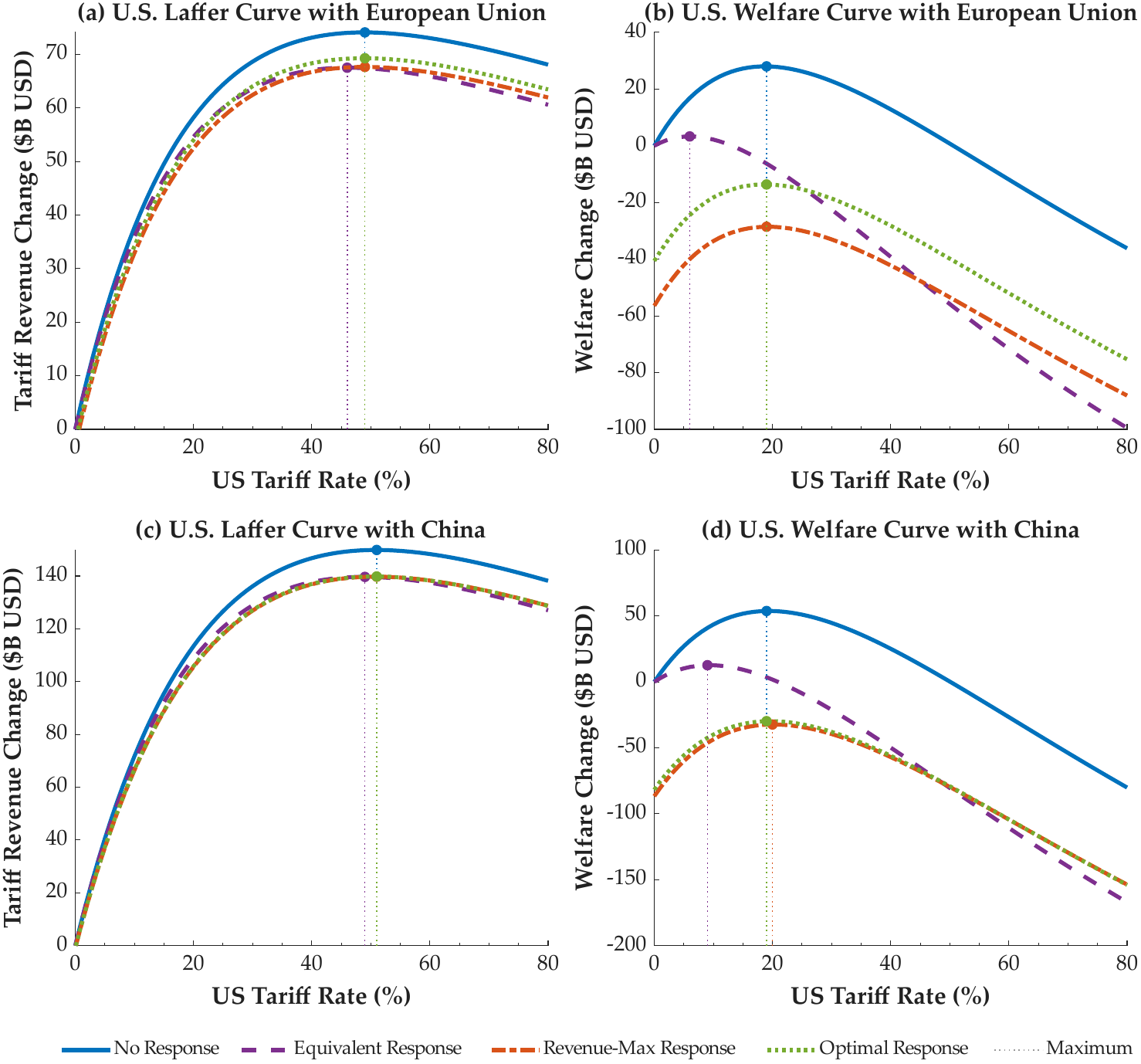}
\end{figure}

 The second hypothesis we test is whether the differences in welfare-maximizing tariff rates could be driven by the dampening effect of input--output linkages, which propagate tariff-induced cost increases through the production network.  We eliminate input-output linkages by calibrating our model using a hypothetical input-output flow matrix where all entries are zeros, i.e. $M_{ki}^s=0$ in equation \eqref{eq:calib_shares} and setting value added to equal gross output in each sector.  This implies that $\gamma_i^s=1$ and $\gamma_i^{k,s}=0$ in equation \eqref{eq:hat_prices}. Figure \ref{fig:Laffer_Sigma24_noIO} shows our results when $\theta^s=2.42$ for all sectors and we eliminate input-output linkages in the model.  The optimal tariff rates are virtually unchanged, except for a slightly lower welfare-maximizing tariff when partners respond equivalently.

\begin{figure}[htbp]
    \centering
    \caption{U.S.\ Laffer Curve with E.U. and China\\($\theta^s=2.42$, No Input-Output Linkages)}
    \label{fig:Laffer_Sigma24_noIO}
    \includegraphics[width=\textwidth]{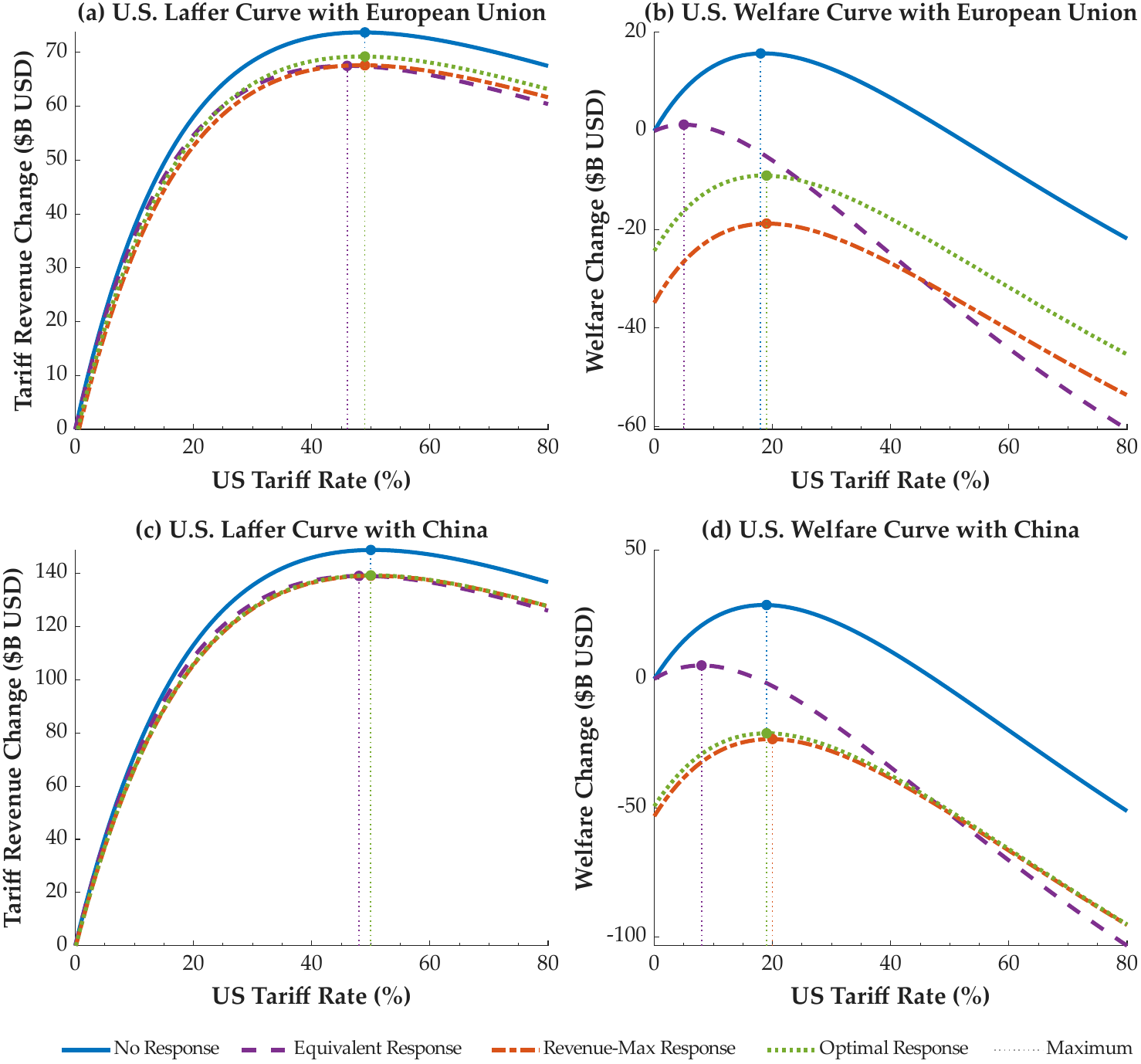}
\end{figure}

Our last experiment more closely mirrors the focus in other papers on bilateral optimality by computing changes in tariff revenue and welfare under the assumption that imports and tariff revenue from other countries are unchanged.  We do this by solving the counterfactual equilibrium with full general equilibrium and third-country effects to find counterfactual changes in imports from the targeted partner, and then we compute tariff revenue and welfare changes using the baseline levels of non-targeted partner trade flows, rather than actual counterfactual levels of these objects.  Figure \ref{fig:Laffer_Sigma24_noIO_bilateral} shows the results of this exercise, and here, the welfare-maximizing tariffs more than double to nearly 42 percent.  Even more notably, when we ignore third-country effects, the welfare-maximizing tariffs no longer respond significantly to equivalent response.  We believe this explains much of the difference between the high Nash Equilibrium tariff rates of around 60 percent found by \citet{Ossa2014} and the much lower tariff rates found in our other paper \citet{pujorossbach2024} which uses a similar framework that fully incorporates all general equilibrium and third-country effects and finds Nash Equilibrium tariff rates closer to ten percent.

\begin{figure}[htbp]
    \centering
    \caption{U.S.\ Laffer Curve with E.U. and China\\($\theta^s=2.42$, No I-O, Bilateral Changes Only)}
    \label{fig:Laffer_Sigma24_noIO_bilateral}
    \includegraphics[width=\textwidth]{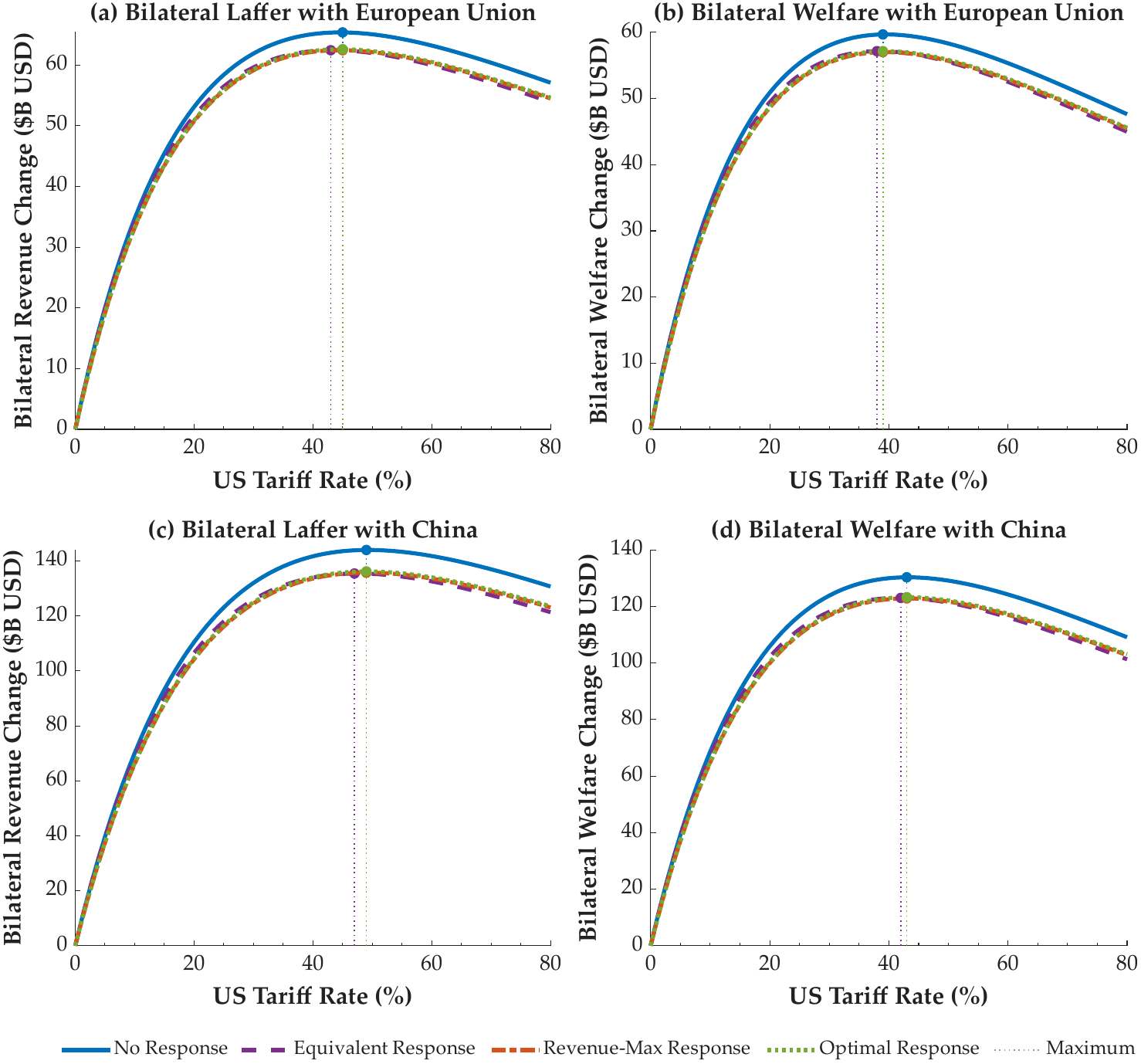}
\end{figure}

\clearpage

\section{Additional Figures and Results}
\label{app:more_results}
In this supplemental appendix, we provide additional results that may be of interest to some readers.  We provide minimal commentary since the results provided here are largely consistent with the results in the other figures and tables or are otherwise explained in the main text.  Table \ref{tbl:laffer_peaks} shows the variation in Laffer peaks and welfare peaks across each U.S.\ trading partner.

\subsection{Robustness of the Inverse-Optimum Exercise}
\label{app-sec:inv_opt_timelines}
Figure \ref{fig:invopt_timeseries} presents the inverse-optimum results separately for each date between 2001--2019 (including the U.S.-China trade war tariffs) and January 2025--2026, with standard bootstrapped confidence intervals due to the lower number of observations compared to pooling dates.  We find the results are largely similar, with the most notable difference being the steep decrease in the total partner weight for USMCA partners in early March 2025, when tariffs were implemented against Canada and Mexico with no exceptions for products that qualified for duty-free access under USMCA until a few days later.  When plotting the revenue weight from the Trade-Off Zone, we plot only points with at least 10 observations in the Trade-Off zone from countries other than USMCA and China.  This is mostly for display purposes, as otherwise there would be a large spike in March 2025 driven by tariff rates for these countries near the Laffer peak, and with an implied revenue weight far above 2.  This outlier spike, exclusive to China and USMCA at the start of the trade war, masks the broader trend of increasing revenue motives of the U.S.\ as applied to a large swath of trading partners.

\begin{figure}[htbp]
    \centering
    \caption{U.S.\ Inverse-Optimum Policy Weights by Date, 2001-January 1, 2026}
    \label{fig:invopt_timeseries}
    \includegraphics[width=\textwidth]{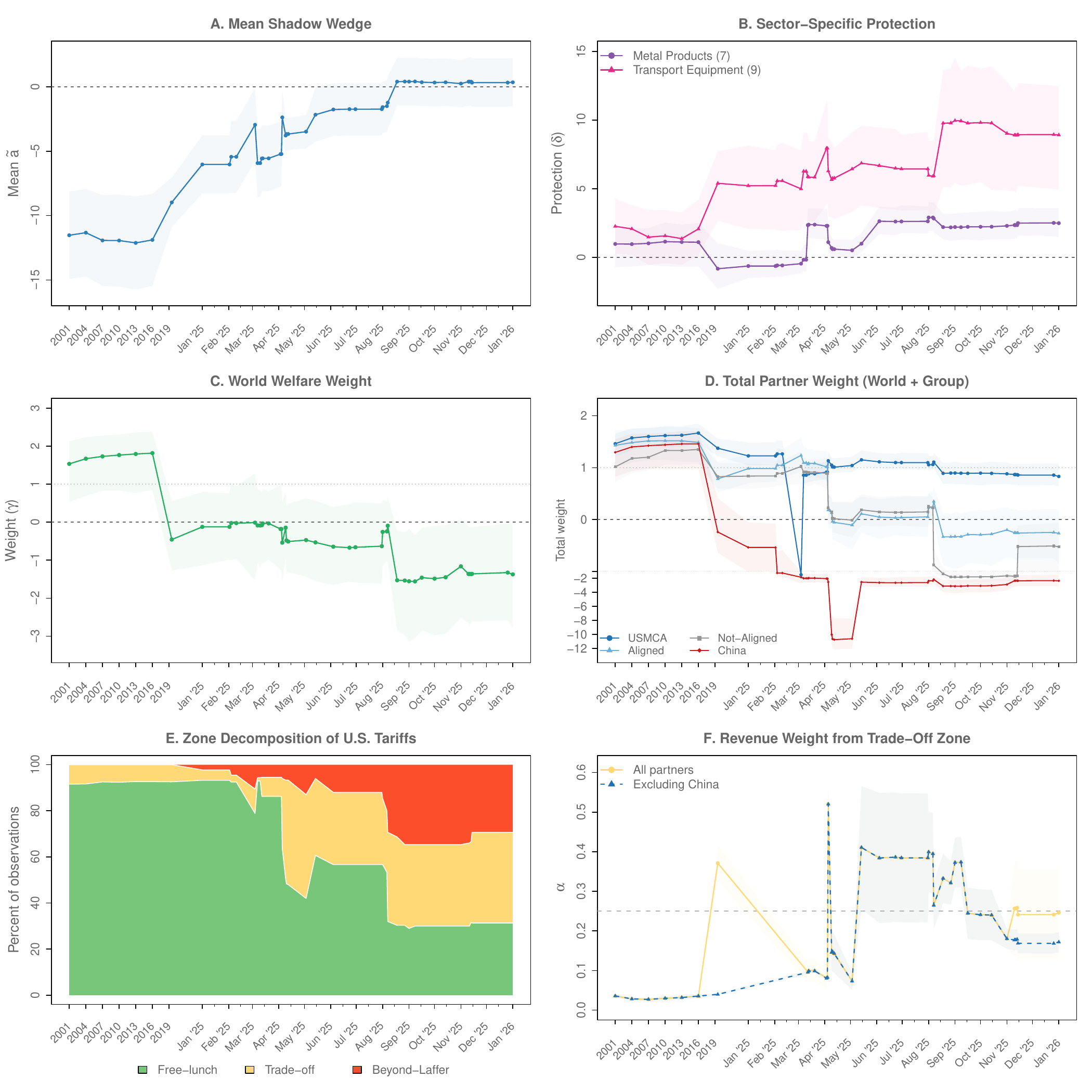}
\end{figure}

\subsubsection{Comparison with Alternative Tariff Data}
\label{app-sec:GTD}
An alternative to our tariff data is the Global Tariff Database (GTD) used by \citet{NBERw33792}, who offer hand-coded tariff rate increases between January and August 2025.  Figure \ref{fig:gtd_vs_wto_panel} displays the mean tariff rate on U.S.\ imports and U.S.\ exports across the GTD and our dataset, the WTO-IMF Tariff Tracker (WTO-TA) database.  Overall, the series largely move in parallel, with the main differences being the longer timeline, lower U.S.\ tariff rates on imports for USMCA in our data, and more variation in U.S.\ export tariff rates in the GTD.  Figure \ref{fig:invopt_timeseries_gtd} presents the global inverse-optimum exercise from section \ref{sec:inverse_optimum} for the full range of dates available in the GTD.  Although bootstrapped confidence intervals are significantly wider with this database, the overall pattern is largely similar.

\begin{figure}[htbp]
    \centering
    \caption{U.S.\ Tariff Rates in GTD vs WTO-TA}
    \label{fig:gtd_vs_wto_panel}
    \includegraphics[width=\textwidth]{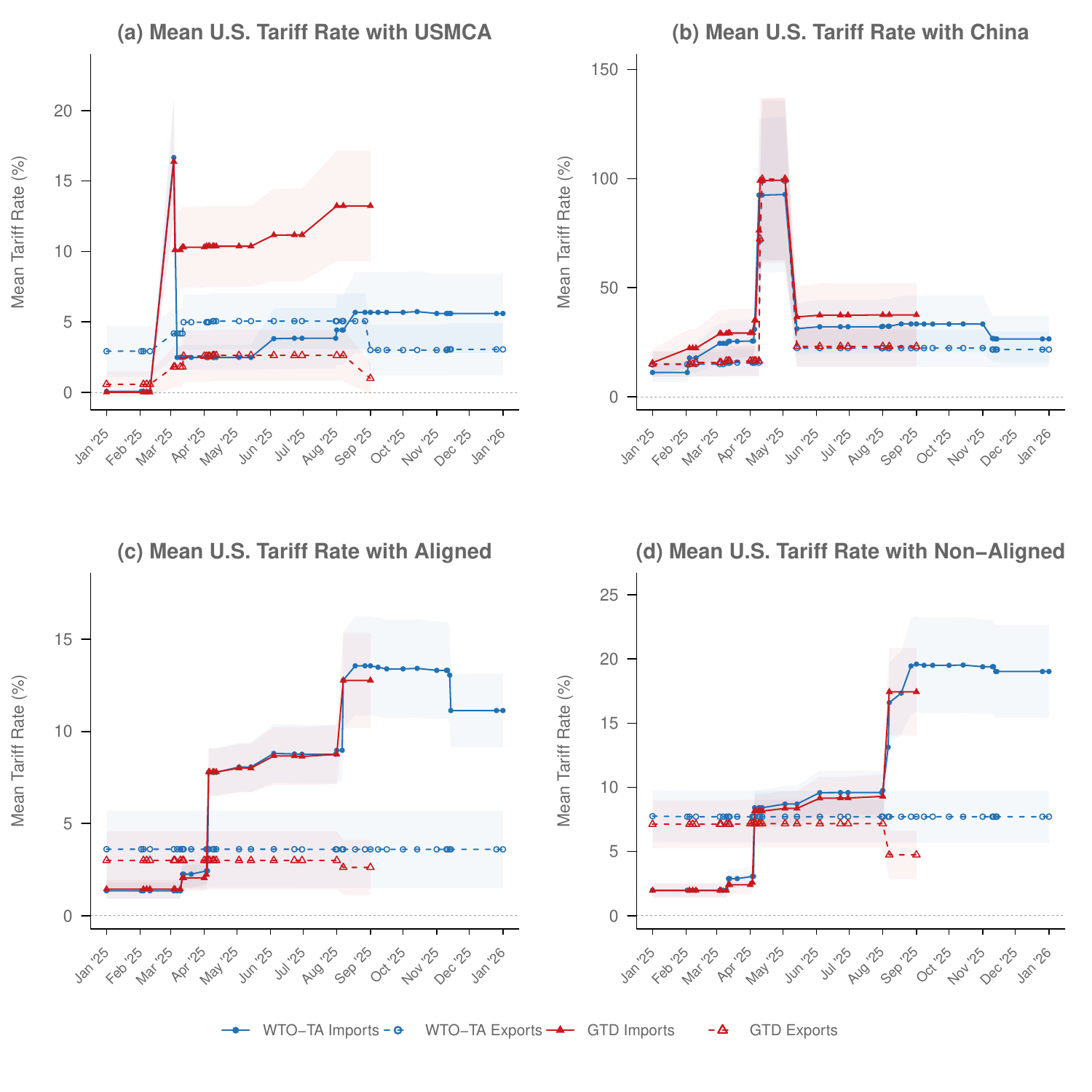}
\end{figure}

\begin{figure}[htbp]
    \centering
    \caption{U.S.\ Inverse-Optimum Policy Weights by Date with GTD}
    \label{fig:invopt_timeseries_gtd}
    \includegraphics[width=\textwidth]{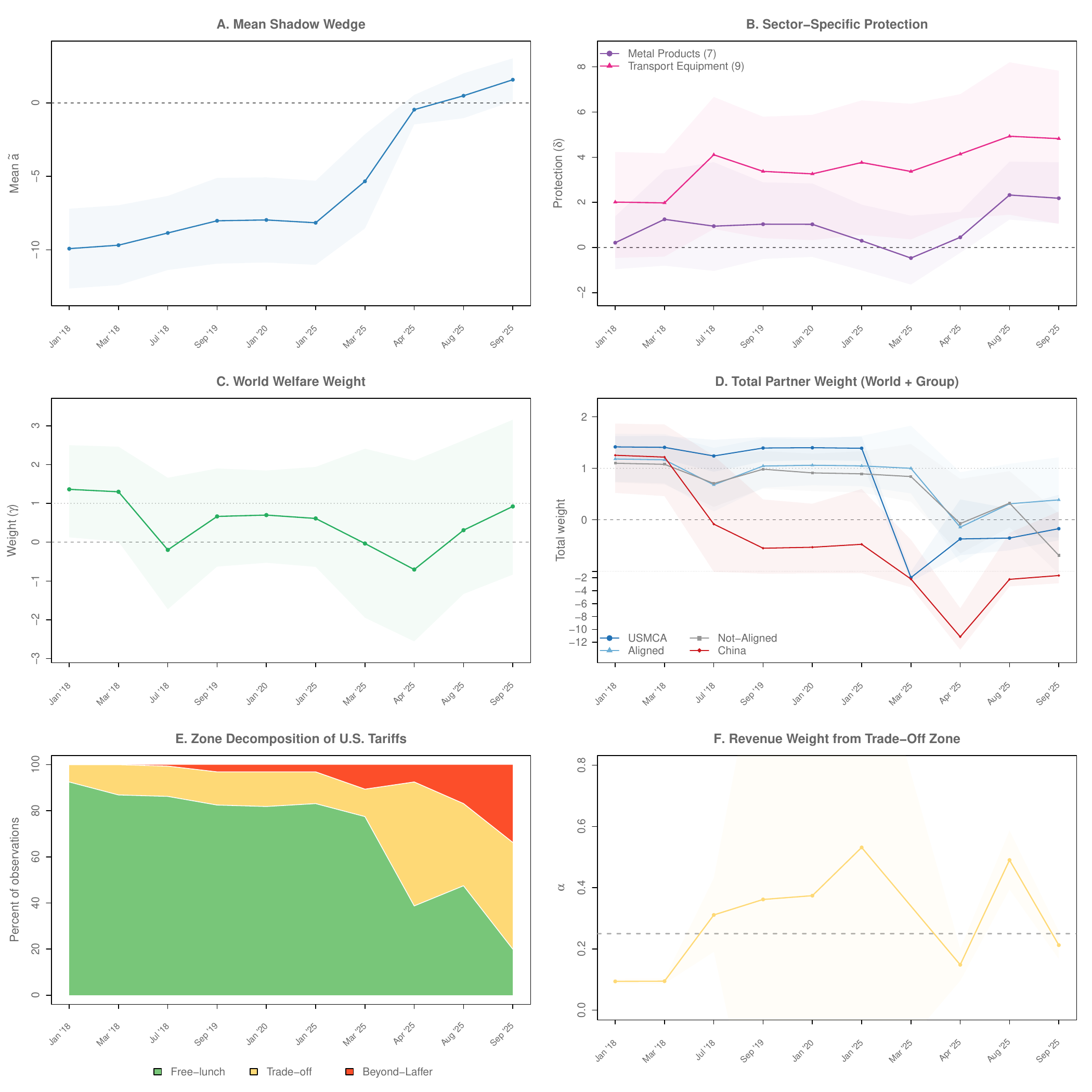}
\end{figure}

\subsubsection{Alternative Weights on Revenue}
\label{app-sec:inv_opt_alpha}
Figure \ref{fig:invopt_alpha_0} presents the first four panels of Figure \ref{fig:invopt_results} with the weight on tariff revenue eliminated by setting $\alpha=0$ (Panels E and F are unaffected by $\alpha$).  Figure \ref{fig:invopt_alpha_50} presents the analog with the weight on tariff revenue increased to $\alpha=0.50$.  Although the wedges and weights shift up and down accordingly as $\alpha$ adjusts, the overall pattern remains largely unchanged.

\begin{figure}[htbp]
    \centering
    \caption{U.S.\ Inverse-Optimum Policy Weights ($\alpha = 0$)}
    \label{fig:invopt_alpha_0}
    \includegraphics[width=\textwidth]{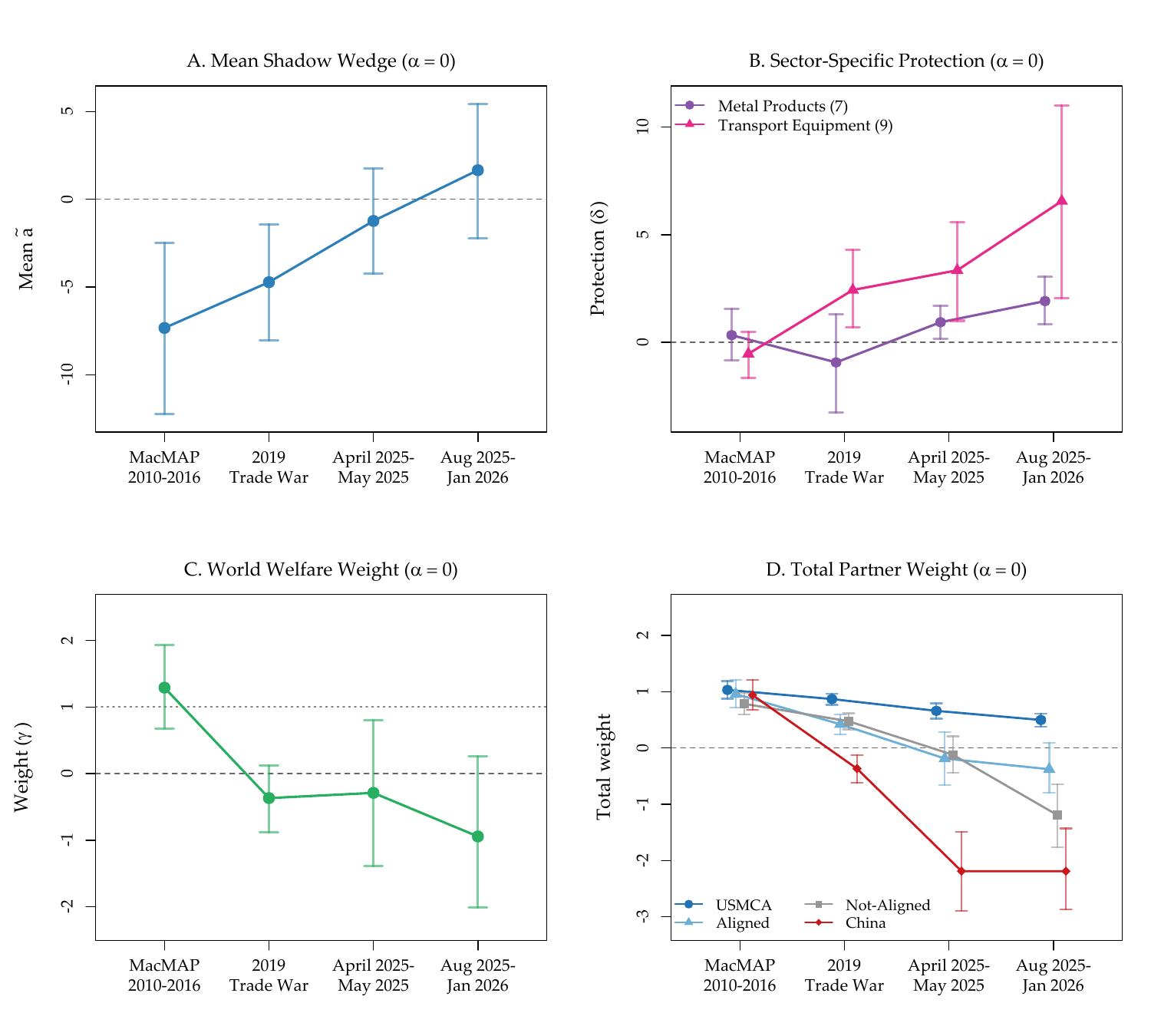}
\end{figure}

\begin{figure}[htbp]
    \centering
    \caption{U.S.\ Inverse-Optimum Policy Weights ($\alpha = 0.50$)}
    \label{fig:invopt_alpha_50}
    \includegraphics[width=\textwidth]{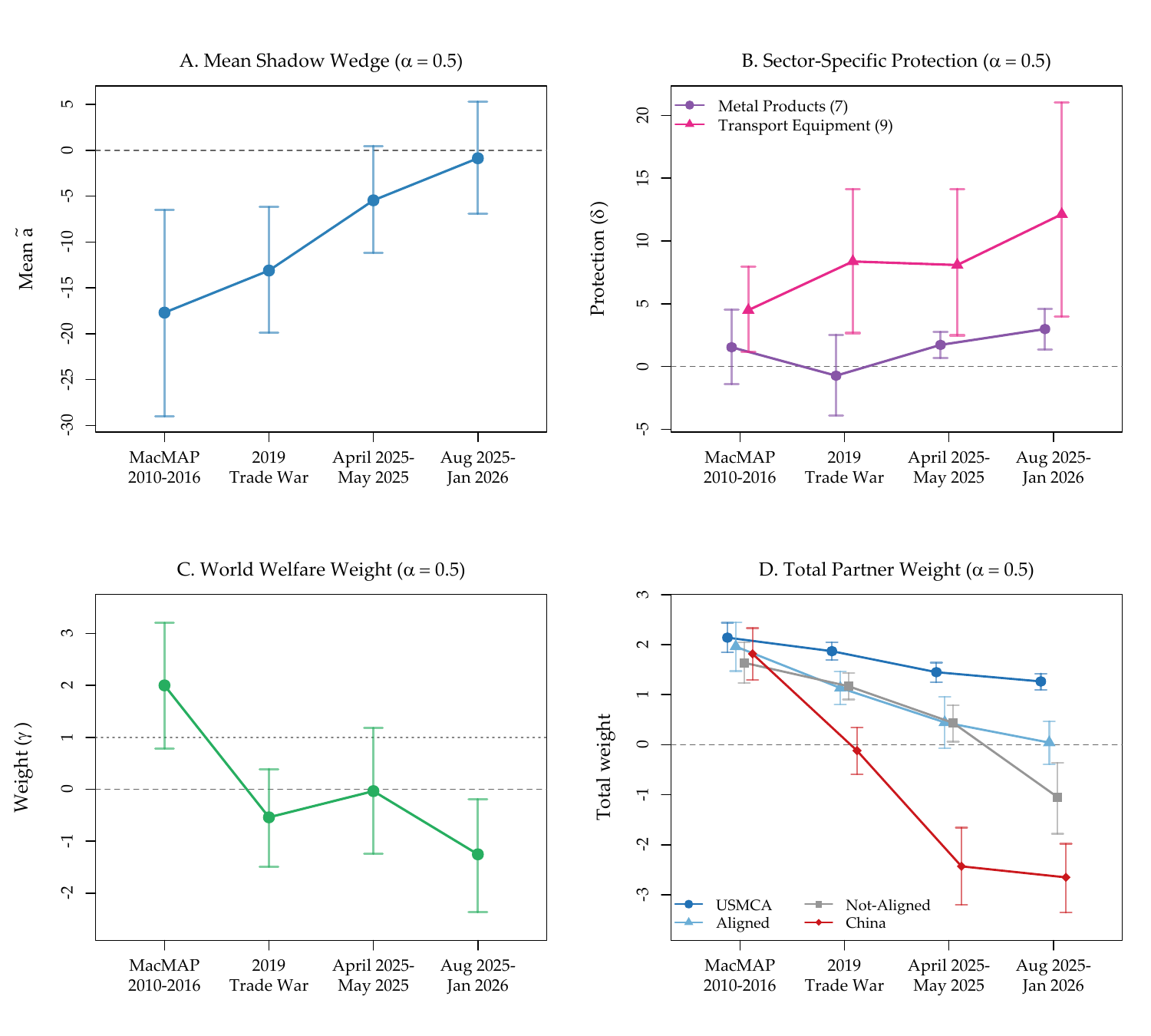}
\end{figure}

\subsection{January 1 Baseline Tariffs}
\label{app:additional_laffers_baselines}
Figure \ref{fig:usa_laffer_eu_chn} shows changes in U.S.\ welfare and tariff revenue from a trade war with the European Union and China across a range of uniform tariff rates and partner responses.  In panel (d) of that figure, the maximum of the Optimal Response curve yields a Nash Equilibrium tariff rate of approximately 9 percent with negative welfare changes for the United States.  This may initially appear to conflict with the results in \citet{pujorossbach2024},  in which we find the United States may experience small welfare gains from a Nash Equilibrium trade war with China using a similar model (in that paper, we solve for optimal tariff rates that vary by sector).  Figure \ref{fig:usa_laffer_eu_chn_Jan1Base} shows changes in U.S.\ welfare and revenue relative to January 1, 2025 baseline tariff rates.  Relative to baseline tariff rates, panel (d) in Figure \ref{fig:usa_laffer_eu_chn_Jan1Base} shows that the Nash Equilibrium tariff rate of 9 percent corresponds to a small welfare increase for the U.S., consistent with the results of our other paper.  These results are in part because baseline tariff rates on January 1, 2025 are not uniform.

\begin{figure}[htbp]
    \centering
    \caption{U.S.\ Laffer and Welfare Curves: European Union and China \\ \emph{\small{Changes Relative to January 1, 2025 Baseline Tariffs}}}
    \label{fig:usa_laffer_eu_chn_Jan1Base}
    \vspace{0.5cm}
    \includegraphics[width=\textwidth]{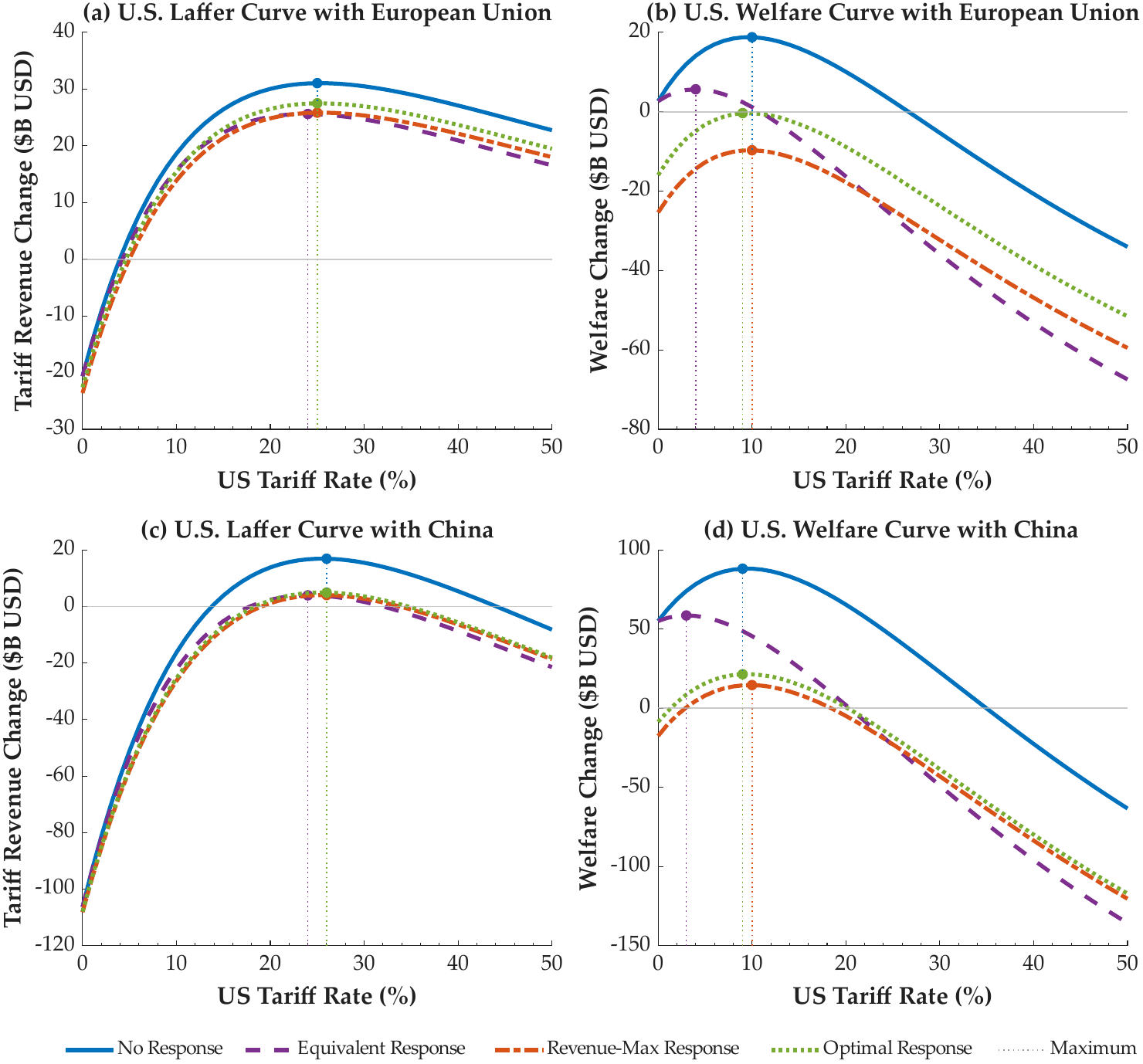}
\end{figure}

\subsection{U.S.\ Laffer Curves with Other Trading Partners}
\label{app:additional_laffers}
Figures \ref{fig:laffer_partners_1} and \ref{fig:laffer_partners_2} present the U.S.\ bilateral tariff revenue Laffer curves and U.S.\ welfare curves for trading partners other than the European Union and China, which are referenced in Section \ref{sec:bilateral_laffer}.  Table \ref{tbl:laffer_peaks} presents the Laffer peak tariff rate and welfare peak tariff rate for each trading partner, as well as the tariff revenue generated by these tariffs.

\begin{figure}[htbp]
    \centering
    \caption{U.S.\ Tariff Laffer Curves and Welfare Curves: Other Partners}
    \label{fig:laffer_partners_1}
    \includegraphics[width=\textwidth]{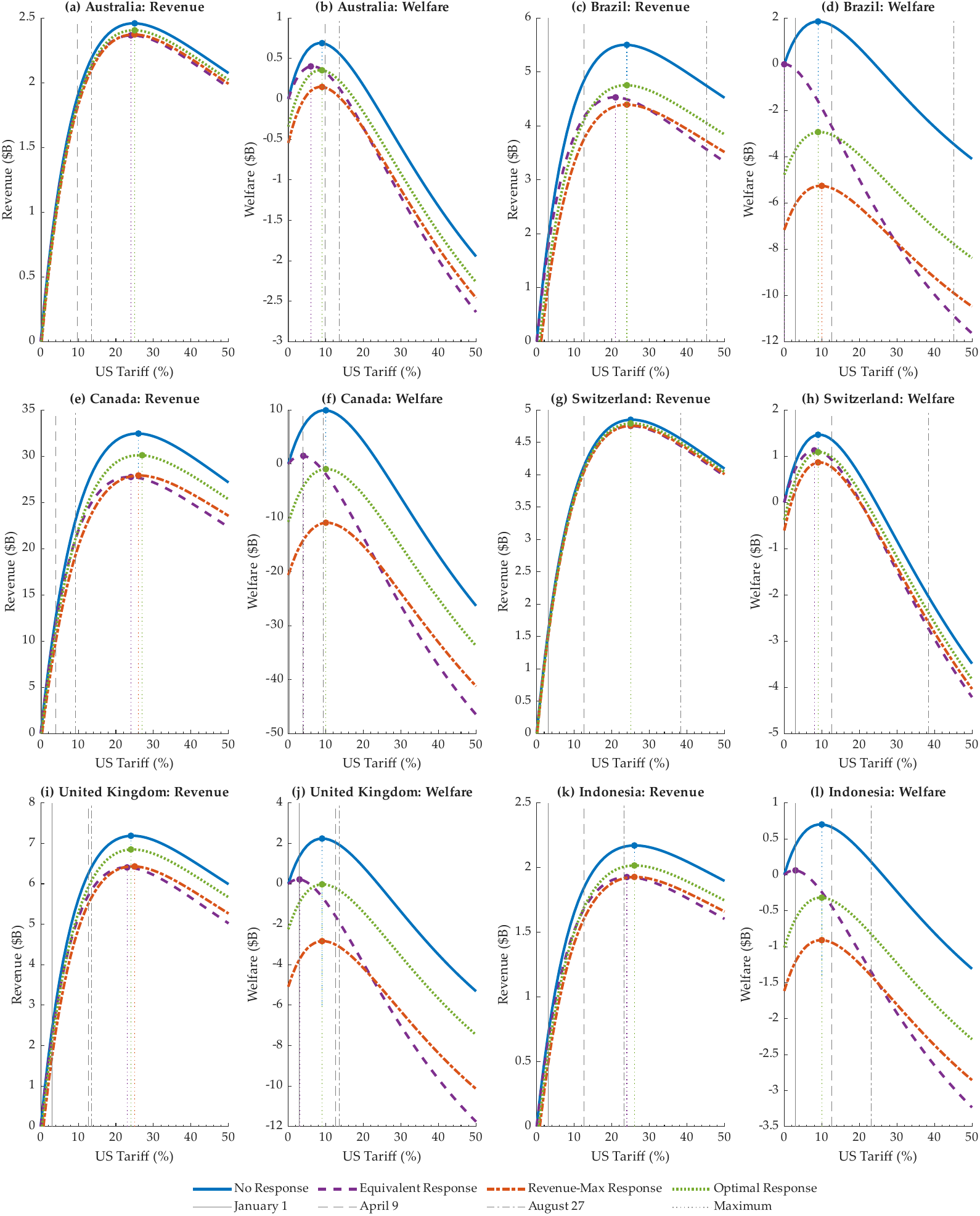}
\end{figure}

\begin{figure}[htbp]
    \centering
    \caption{U.S.\ Tariff Laffer Curves and Welfare Curves: Remaining Partners}
    \label{fig:laffer_partners_2}
    \includegraphics[width=\textwidth]{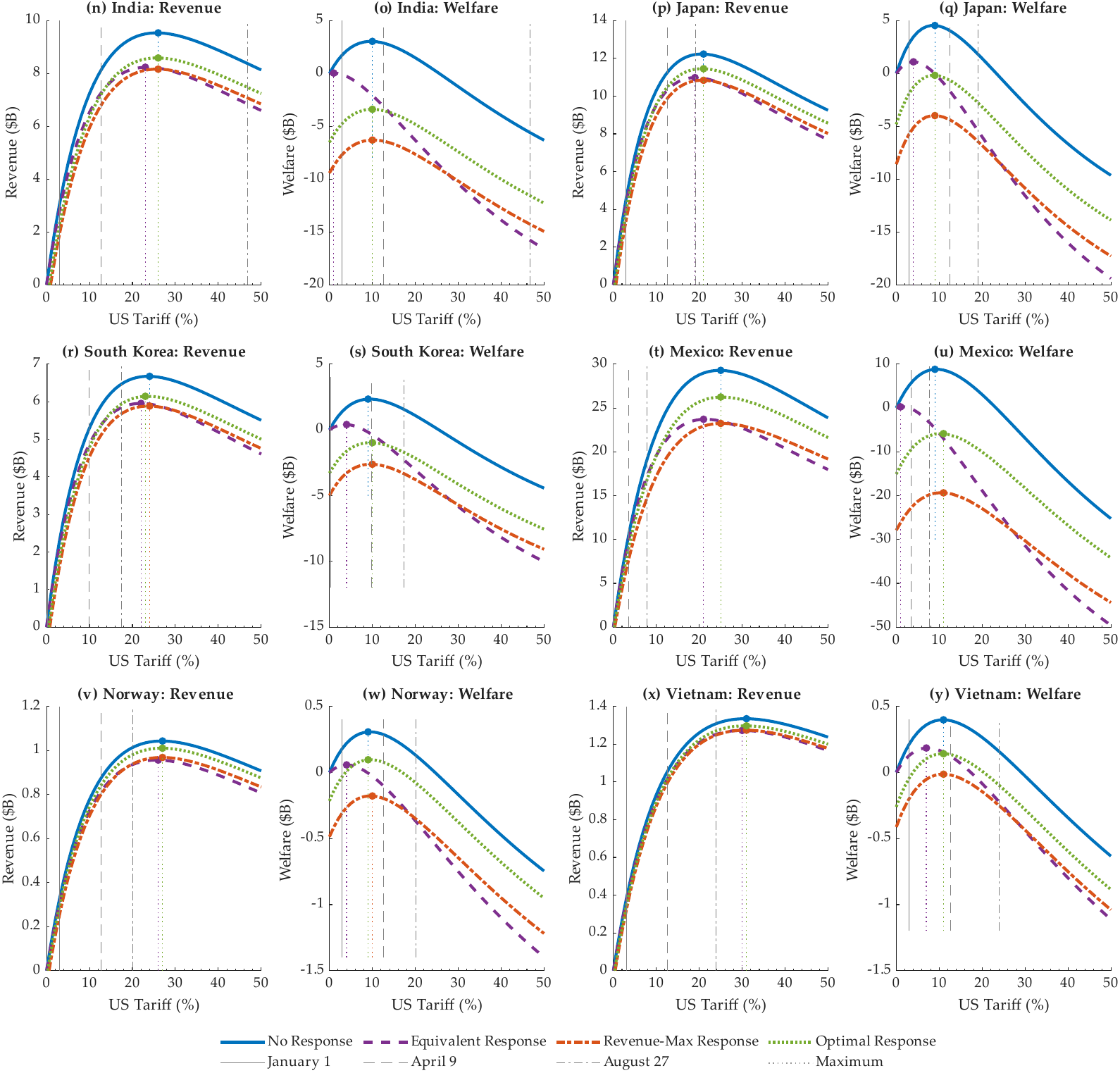}
\end{figure}

\begin{table}[h!]
\centering
\caption{Laffer Curve Peaks: No Retaliation vs Equivalent Retaliation}
\label{tbl:laffer_peaks}
\small
\begin{tabular}{lccccc|ccccc}
\toprule
& \multicolumn{5}{c}{\textbf{No Retaliation}} & \multicolumn{5}{c}{\textbf{Equivalent Retaliation}} \\
\cmidrule(lr){2-6} \cmidrule(lr){7-11}
& \textbf{Welfare} & \textbf{Laffer} & \multicolumn{3}{c}{\underline{\textbf{At US Laffer Peak (\$B)}}} & \textbf{Welfare} &  \textbf{Laffer}& \multicolumn{3}{c}{\underline{\textbf{At US Laffer Peak (\$B)}}} \\
 & \textbf{Peak} & \textbf{Peak} & \textbf{US} & \textbf{US} & \textbf{Ptnr} & \textbf{Peak} & \textbf{Laffer} & \textbf{US} & \textbf{US} & \textbf{Ptnr} \\
\textbf{Partner}& \textbf{$\tau$(\%)} & \textbf{$\tau$ (\%)} & \textbf{Rev} & \textbf{Wel.} & \textbf{Wel.} & \textbf{$\tau$(\%)} & \textbf{$\tau$(\%)} & \textbf{Rev.} & \textbf{Wel.} & \textbf{Wel.} \\
\midrule
AUS & 9 & 25 & 2.5 & -0.2 & -2.4 & 6 & 24 & 2.4 & -0.7 & -2.4 \\
BRA & 9 & 24 & 5.5 & 0.0 & -7.9 & 0 & 21 & 4.5 & -5.3 & -4.2 \\
CAN & 10 & 26 & 32.5 & -1.9 & -33.1 & 4 & 24 & 27.8 & -18.8 & -33.4 \\
CHE & 9 & 25 & 4.8 & -0.1 & -5.1 & 8 & 25 & 4.8 & -0.7 & -5.0 \\
CHN & 9 & 26 & 123.2 & -14.3 & -238.2 & 3 & 24 & 110.3 & -73.3 & -151.9 \\
EUU & 10 & 25 & 51.7 & -0.0 & -59.7 & 4 & 24 & 46.3 & -26.9 & -52.0 \\
GBR & 9 & 24 & 7.2 & -0.0 & -6.7 & 3 & 23 & 6.4 & -4.9 & -7.8 \\
IDN & 10 & 26 & 2.2 & 0.0 & -2.8 & 3 & 24 & 1.9 & -1.4 & -2.3 \\
IND & 10 & 26 & 9.5 & -0.1 & -12.8 & 1 & 23 & 8.2 & -7.6 & -8.2 \\
JPN & 9 & 21 & 12.2 & 1.0 & -14.2 & 4 & 19 & 11.0 & -5.4 & -12.4 \\
KOR & 9 & 24 & 6.7 & 0.3 & -9.1 & 4 & 22 & 6.0 & -3.6 & -6.8 \\
MEX & 9 & 25 & 29.3 & -1.8 & -27.5 & 1 & 21 & 23.7 & -20.2 & -24.1 \\
NOR & 9 & 27 & 1.0 & -0.1 & -1.4 & 4 & 26 & 1.0 & -0.6 & -1.6 \\
TUR & 10 & 27 & 2.6 & 0.0 & -3.1 & 6 & 26 & 2.4 & -0.9 & -3.2 \\
VNM & 11 & 31 & 1.3 & -0.1 & -1.6 & 7 & 30 & 1.3 & -0.4 & -1.5 \\
\bottomrule
\end{tabular}
\end{table}

\clearpage
\subsection{MFEI for Optimal Tariffs}
\label{sec-app:MFEI_opt_response}
Figure \ref{fig:MEB_opt_response} presents the Marginal Excess Burden and Marginal Fiscal Efficiency Index from section \ref{sec:MEB}. Figure \ref{fig:MEB_no_response_alpha50} presents the counterpart to Figure \ref{fig:MEB_no_response}, with $\alpha=0.5$.  Note that increasing $\alpha$ shifts the MFEI curve to the left within the Trade-Off zone, as revenue becomes more highly valued relative to welfare; however, the Free-Lunch and Beyond-Laffer zones are not affected.

\begin{figure}[htbp]
    \centering
    \caption{Marginal Welfare Cost Measures of U.S.\ Tariff Revenue \\(with Optimal Partner Response)}
    \label{fig:MEB_opt_response}
    \includegraphics[width=\textwidth]{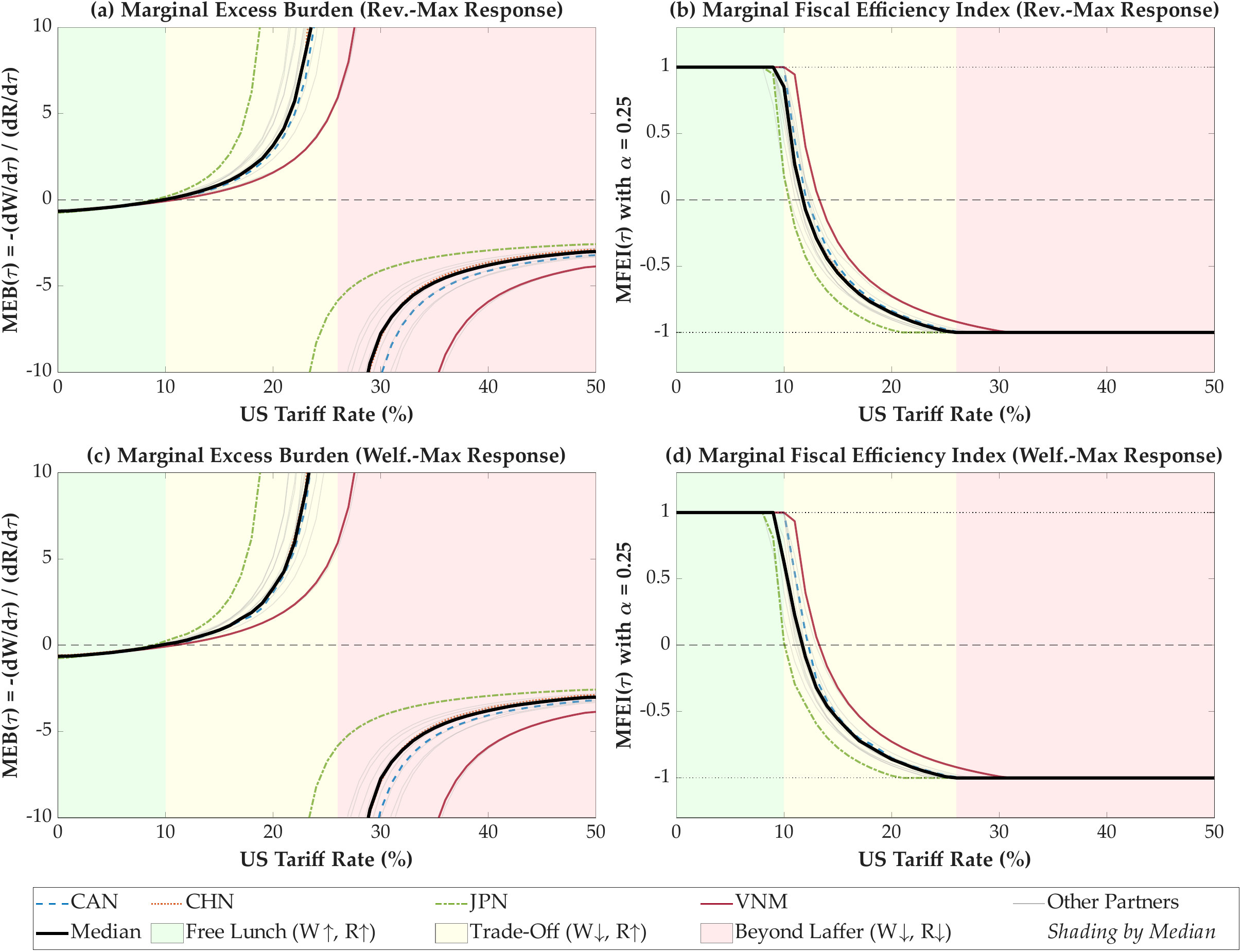}
\end{figure}

\begin{figure}[htbp]
    \centering
    \caption{Marginal Welfare Cost Measures of U.S.\ Tariff Revenue (with $\alpha = 0.50$)}
    \label{fig:MEB_no_response_alpha50}
    \includegraphics[width=\textwidth]{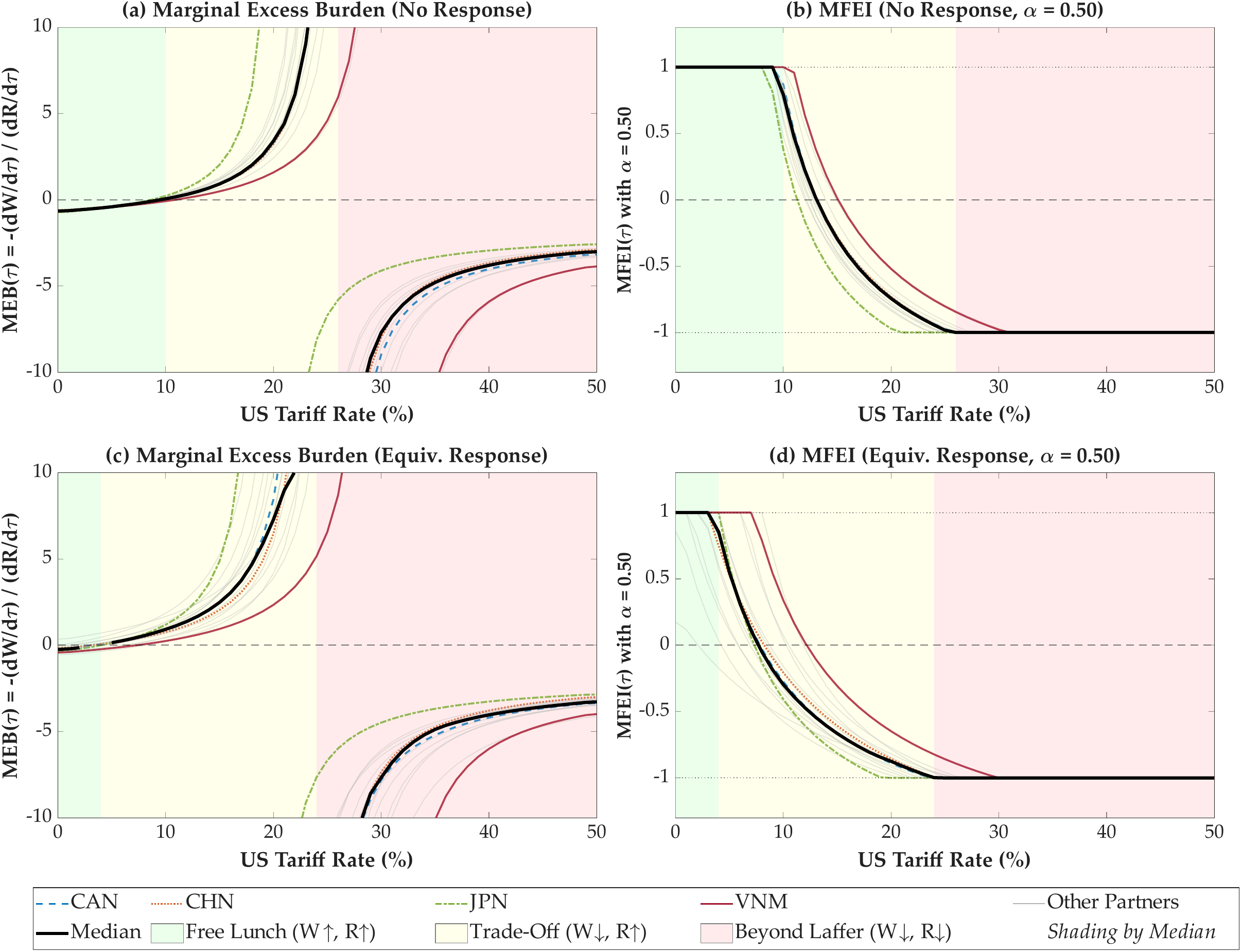}
\end{figure}

\subsection{Cumulative Laffer Curves with Optimal Partner Responses}
\label{sec-app:cumulative_opt_response}
We evaluate the Laffer curve for the U.S.\ in a trade war with multiple countries, assuming that partner countries respond with uniform tariffs on imports from the U.S.\ that maximize their own welfare or revenue.  For each U.S.\ tariff level, partners choose their retaliatory tariffs sequentially in a fixed order, with each partner maximizing its own welfare taking as given the U.S. tariff rate and the tariff rates of all partners that have already moved. This is a one-pass procedure, where earlier movers do not re-optimize after later partners choose, and the resulting tariff profile represents a sequential best-response path rather than a simultaneous-move Nash equilibrium of the multi-partner game.  We verified that a second pass has a negligible effect on optimal response rates for a trade war involving all U.S.\ trading partners, and that the cumulative Laffer curves are virtually unchanged.  In these exercises, partners are not constrained to place the same uniform tariff rate on the United States, and we assume that each partner takes into account only their own welfare or revenue when setting their tariff.

A first-order effect of adding partners to the Trade War is that they raise retaliatory tariffs even when the U.S.\ tariff remains at zero.  Figure~\ref{fig:cum_Laffer_4panel_adaptive} isolates the effect of U.S.\ tariff rate increases on a trade war involving the world, by showing changes relative to the counterfactual trade war with U.S.\ tariff rates at zero.  Figure~\ref{fig:cum_Laffer_4panel_adaptive_ftBase} displays cumulative Laffer and welfare curves when partners respond with welfare-maximizing tariffs and with revenue-maximizing tariffs, with all changes reported relative to a U.S.\ Free Trade Baseline with all partners; and Figure~\ref{fig:cum_Laffer_4panel_adaptive_Jan1Base} shows the same results relative to January 1, 2025 baseline tariffs.  Note that the effect of a marginal 1 percent increase in U.S.\ tariff rates is the same in each of these three figures.  The dotted line corresponding to the maximum for the Welfare-Maximizing Response is an approximation of the Nash Equilibrium.  These results show that although a global trade war involving all trading partners is worse for the U.S.\ than a free-trade equilibrium, in general, welfare-maximizing responses by individual partners do not appear to be a strong deterrent relative to January 1, 2025 baseline tariffs.

\begin{figure}[htbp]
    \centering
    \caption{U.S.\ Laffer and Welfare Curves: Multiple Trading Partners \\ \emph{\small{Changes Relative to Partners' Optimal Response with U.S.\ Tariff Rate of Zero}}}
    \label{fig:cum_Laffer_4panel_adaptive}
    \includegraphics[width=\textwidth]{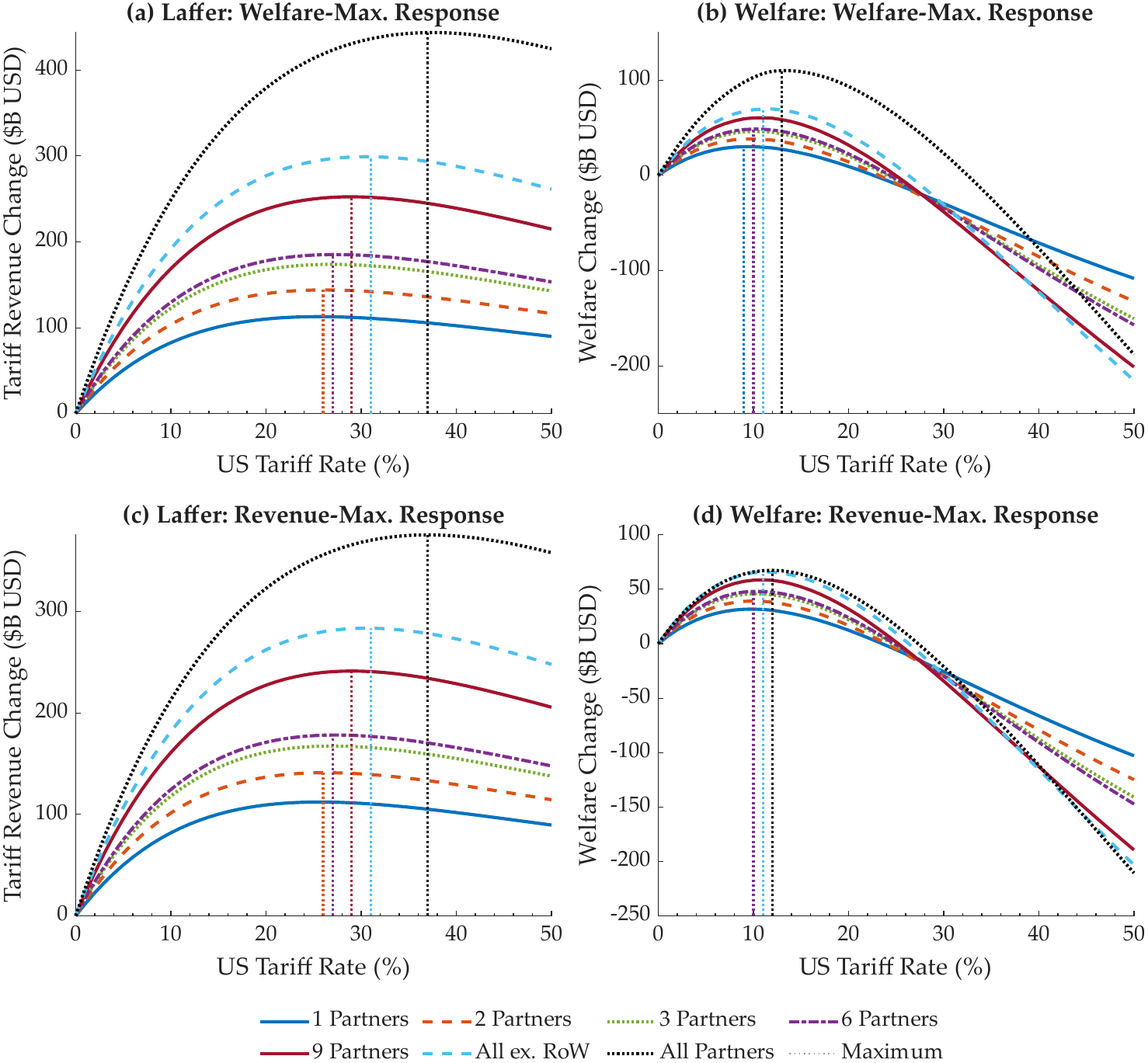}
    \vspace{0.5em}
    \begin{minipage}{\textwidth}
        \begin{singlespacing}
            \fontsize{8.5}{10}\selectfont
            \noindent\textit{Notes:} ``\# Partners'' indicates the number of partners on which the U.S.\ imposes tariffs and from which it potentially faces retaliation. Partner order: (1) CHN (2) CAN (3) MEX (4) VNM (5) TUR (6) KOR (7) IND (8) IDN (9) EUU (10) BRA (11) NOR (12) CHE (13) JPN (14) GBR (15) AUS (16) Rest of World.   Results robust to alternative orderings.
        \end{singlespacing}
    \end{minipage}
\end{figure}

\begin{figure}[htbp]
    \centering
    \caption{U.S.\ Laffer and Welfare Curves: Multiple Trading Partners \\ \emph{\small{Changes Relative to U.S.\ Free Trade Baseline with All Partners}}}
    \label{fig:cum_Laffer_4panel_adaptive_ftBase}
    \includegraphics[width=\textwidth]{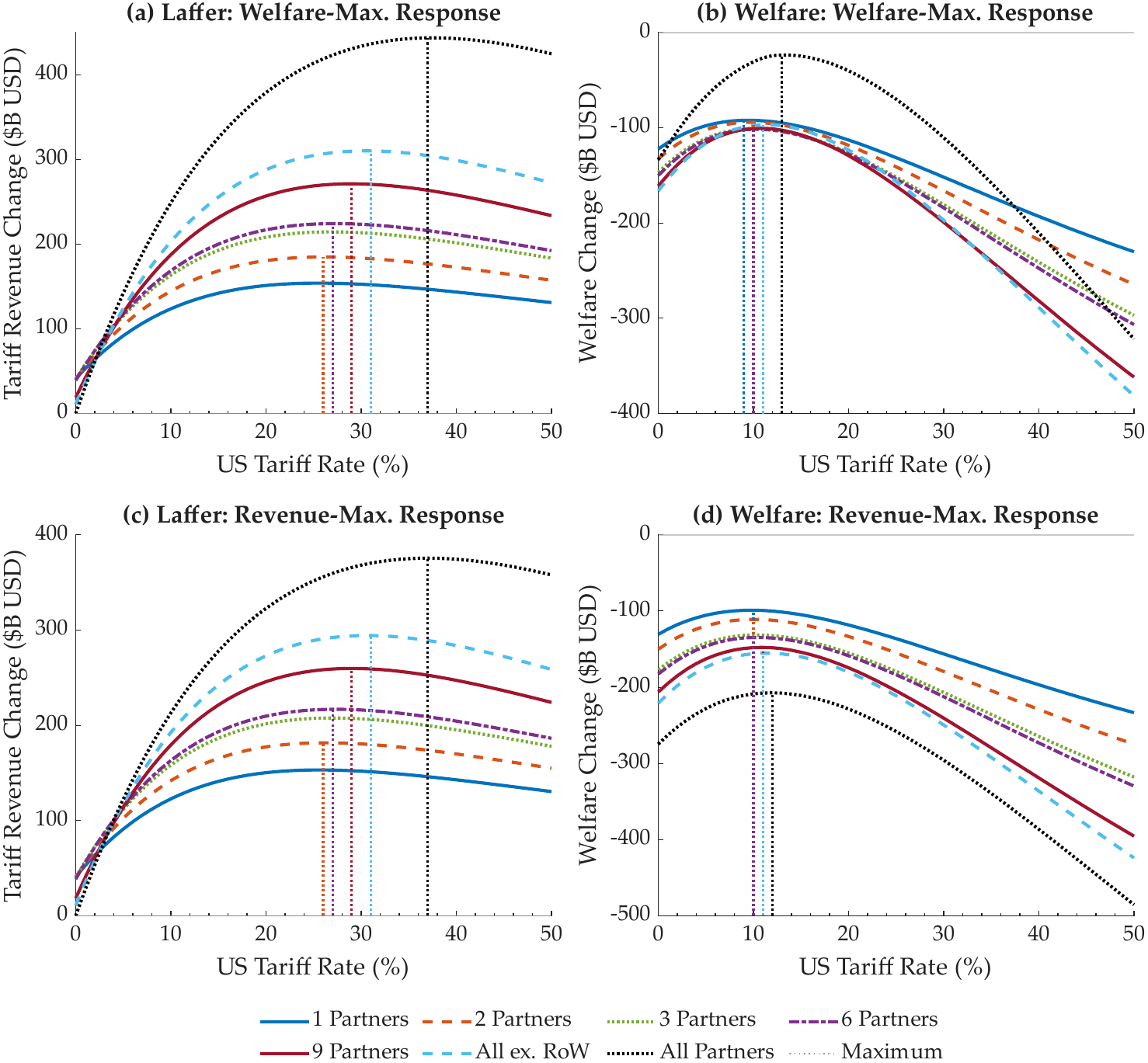}
    \vspace{0.5em}
    \begin{minipage}{\textwidth}
        \begin{singlespacing}
            \fontsize{8.5}{10}\selectfont
        \end{singlespacing}
    \end{minipage}
\end{figure}

\begin{figure}[htbp]
    \centering
    \caption{U.S.\ Laffer and Welfare Curves: Multiple Trading Partners \\ \emph{\small{Changes Relative to January 1, 2025 Baseline Tariff Rates}}}
    \label{fig:cum_Laffer_4panel_adaptive_Jan1Base}
    \includegraphics[width=\textwidth]{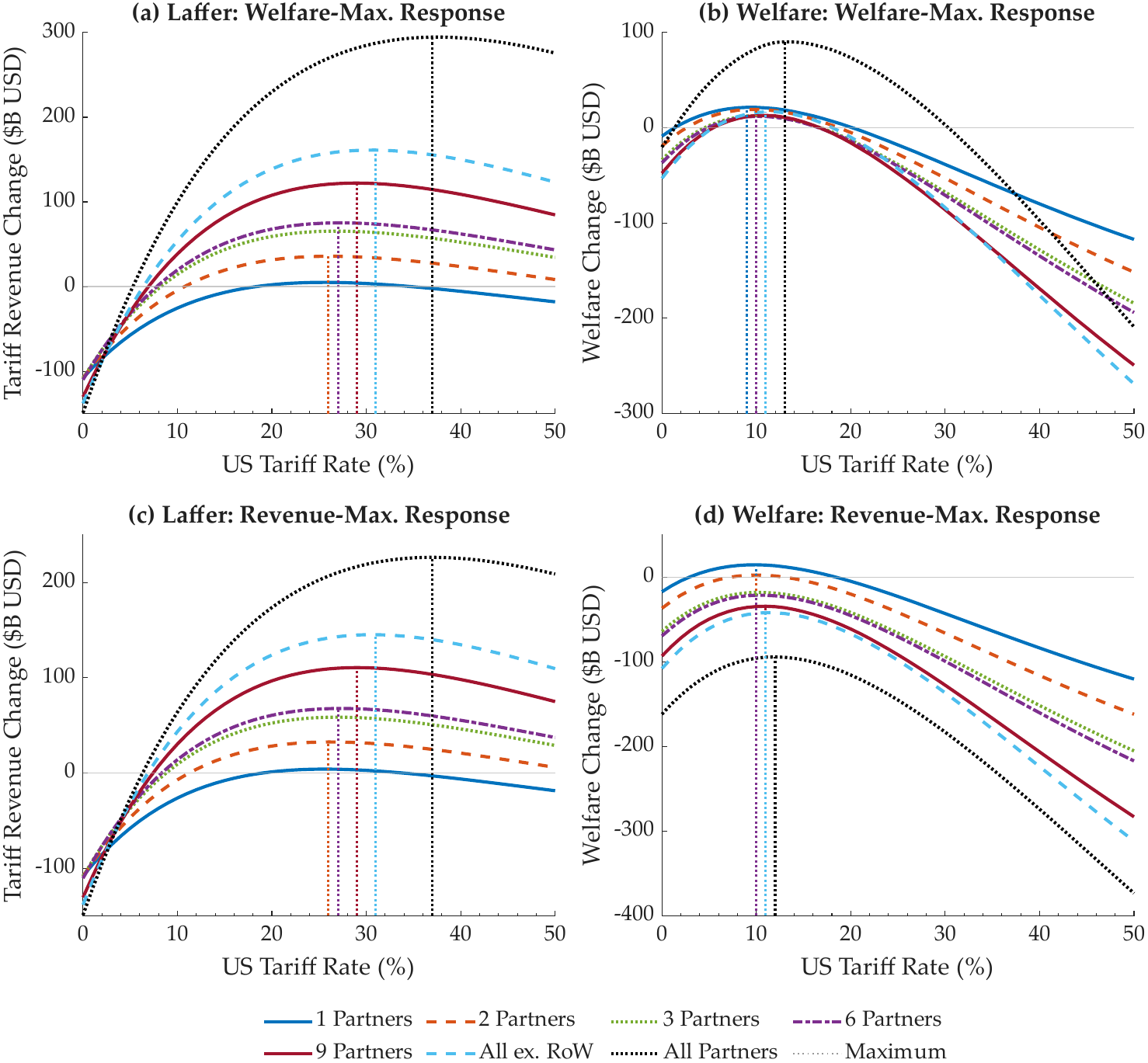}
    \vspace{0.5em}
    \begin{minipage}{\textwidth}
        \begin{singlespacing}
            \fontsize{8.5}{10}\selectfont
        \end{singlespacing}
    \end{minipage}
\end{figure}

\clearpage
\section{Calibration and Computational Details}
\label{app:calibration_details}

This appendix documents the conventions, parameter values, and computational details underlying the calibration and quantitative exercises.

\subsection{Country Aggregation}

The model contains $N=17$ regions. Eora26 reports individual countries; we aggregate as follows.  EU-27 member states (Austria, Belgium, Bulgaria, Croatia, Cyprus, Czech Republic, Denmark, Estonia, Finland, France, Germany, Greece, Hungary, Ireland, Italy, Latvia, Lithuania, Luxembourg, Malta, Netherlands, Poland, Portugal, Romania, Slovakia, Slovenia, Spain, Sweden) are summed into a single region.  We aggregate bilateral trade flows and input-output flows by summing values across countries within a region.  Table~\ref{tbl:regions} lists the 17 regions.

\begin{table}[htbp]
\centering
\caption{Model Regions}
\label{tbl:regions}
\begin{small}
\begin{tabular}{clcl}
\toprule
Index & Region & Index & Region \\
\midrule
1 & United States (USA) & 10 & Turkey (TUR) \\
2 & Canada (CAN) & 11 & Switzerland (CHE) \\
3 & Mexico (MEX) & 12 & Norway (NOR) \\
4 & China (CHN) & 13 & Brazil (BRA) \\
5 & European Union, EU-27 (EUU) & 14 & India (IND) \\
6 & United Kingdom (GBR) & 15 & Indonesia (IDN) \\
7 & Japan (JPN) & 16 & Vietnam (VNM) \\
8 & South Korea (KOR) & 17 & Rest of World (ROW) \\
9 & Australia (AUS) & & \\
\bottomrule
\end{tabular}
\end{small}
\end{table}

\subsection{Sector Aggregation}

We aggregate the 26 Eora sectors into $S = 15$ model sectors: 10 tariffable tradable-goods sectors and 5 non-tariffable, but still tradable, service sectors.  Eora sector~26 (Re-exports and Re-imports) is dropped.  Table~\ref{tbl:sectors} reports the mapping.

\begin{table}[htbp]
\centering
\caption{Sector Aggregation: Eora26 to Model Sectors}
\label{tbl:sectors}
\begin{small}
\begin{tabular}{clccc}
\toprule
$s$ & Name & Eora Codes & ISIC Rev.~3 & Tariffable \\
\midrule
1 & Agriculture \& Fishing & 1, 2 & 1, 2, 5 & Yes \\
2 & Mining \& Quarrying & 3 & 10--14 & Yes \\
3 & Food \& Beverages & 4 & 15, 16 & Yes \\
4 & Textiles \& Wearing Apparel & 5 & 17--19 & Yes \\
5 & Wood \& Paper & 6 & 20--22 & Yes \\
6 & \shortstack{Petroleum, Chemicals, \& Minerals} & 7 & 23--26 & Yes \\
7 & Metal Products & 8 & 27, 28 & Yes \\
8 & Electrical \& Machinery & 9 & 29--33 & Yes \\
9 & Transport Equipment & 10 & 34, 35 & Yes \\
10 & Other Manufacturing incl.\ Recycling & 11, 12 & 36, 37 & Yes \\
\midrule
11 & Utilities \& Construction & 13, 14 & 40, 41, 45 & No \\
12 & Trade \& Hospitality & 15--18 & 50--52, 55 & No \\
13 & Transport \& Communication & 19, 20 & 60--64 & No \\
14 & Finance \& Business Services & 21 & 65--67, 70--74 & No \\
15 & Public, Personal, \& Other Services & 22--25 & 75, 80, 85, 90--93, 95, 99 & No \\
\bottomrule
\end{tabular}
\end{small}
\end{table}

\subsection{Trade Elasticities}

Table~\ref{tbl:elasticities} reports the sector-specific trade elasticities.  All sectoral elasticities are from \citet{Chen2023Fragmentation}.  

\begin{table}[htbp]
\centering
\caption{Trade Elasticities by Sector}
\label{tbl:elasticities}
\begin{small}
\begin{tabular}{clr}
\toprule
$s$ & Name & $\theta^s$ \\
\midrule
1 & Agriculture and Fishing & 2.910 \\
2 & Mining and Quarrying & 3.410 \\
3 & Food and Beverages & 4.170 \\
4 & Textiles and Wearing Apparel & 4.710 \\
5 & Wood and Paper & 8.505 \\
6 & Petroleum, Chemicals, Non-Metallic Min. & 6.443 \\
7 & Metal Products & 5.805 \\
8 & Electrical and Machinery & 4.753 \\
9 & Transport Equipment & 8.955 \\
10 & Other Manufacturing incl.\ Recycling & 4.060 \\
11--15 & All service sectors & 8.350 \\
\bottomrule
\end{tabular}
\end{small}
\end{table}

\subsection{Tariff Construction}
Notation: $t_{ij}^s$ denotes the ad-valorem tariff rate (e.g., $t = 0.25$ for a 25\% tariff) and $\tau_{ij}^s = 1 + t_{ij}^s$ denotes the gross tariff factor.  Subscript $i$ denotes the destination and $j$ the source, so $\tau_{ij}^s$ is the tariff imposed by $i$ on imports from $j$ in sector $s$.

Baseline source-destination-sector tariff rates are assembled from three sources.  Tariff rates are collected for 6-digit Harmonized System (HS) product codes in the referenced revision for each partner.  
\begin{enumerate}
\item \textbf{MacMap baseline.} Bilateral ad-valorem rates from the CEPII MacMap-HS6 database.  We use the 2019 revision, which contains tariffs every three years from 2001--2019.  The concordance chain is: HS2007 \(\to\)\footnote{Product to Industry concordances are from WITS: \url{https://wits.worldbank.org/product_concordance.html}}   ISIC Rev.~3 (2-digit) \(\to\) Eora sector from Table \ref{tbl:sectors}. Within each source--destination--sector group, the simple average is taken across HS6 products.
\item \textbf{2018--19 trade war.} Tariff rate increases for U.S.\ imports and exports from \citet{Fajgelbaum2019}. The concordance
  chain for these data is: HS2017 \(\to\) HS2012 \(\to\) ISIC Rev.~3 (2-digit) \(\to\) Eora
  sector. Within each group, the simple average increase is calculated as $\Delta t_{ij}^{s,\text{RtP}}$.  Tariffs for 2019 with the Trade War are added to the 2019 baseline: $t_{ij}^{s} = t_{ij}^{s,\text{MacMap}} + \Delta t_{ij}^{s,\text{RtP}}$.
\item \textbf{WTO-IMF Tariff Tracker.} Where the Tracker provides a January~1, 2025 rate for a bilateral-sector pair, it replaces the MacMap$+$trade-war rate.  The concordance chain for the WTO data is: HS2022 \(\to\)
  HS2012 \(\to\) ISIC Rev.~3 (2-digit)
  \(\to\) Eora sector. For each policy date \(d\) between January 1, 2025 and January 1, 2026, tariffs through January~1, 2026 are applied chronologically. 
\end{enumerate}
Self-tariff rates, $i=j$, and service-sector tariff rates are set to zero ($\tau_{ij}^s=1$).

\subsection{Unit Conventions}
\label{app-sec:units}
All Eora values are in billions of current-year USD.  Therefore, our baseline calibration uses values in 2023 USD.  Dollar amounts reported in the paper are scaled to 2025 USD by multiplying by $\kappa = 1.106$, the mean ratio of quarterly current-price U.S.\ GDP in 2023 to quarterly current-price U.S.\ GDP in 2025 (FRED series GDP, only Q1--Q3 were available at time of these exercises).  Welfare changes for region~$j$ are calculated as $(\hat{W}_j - 1) \times \text{GDP}_j$, where $\text{GDP}_j$ is total baseline value added in the calibrated model, converted to 2025 USD using the same scaling.  Welfare changes for a collection of regions is computed as the sum of changes across each region in the collection.  This is equivalent to a GDP-weighted average change, multiplied by the total baseline GDP.

\subsection{Quantitative Exercises}
\label{app:exercises}

This section provides additional details regarding our quantitative exercise.  All exercises use the model described in Section~\ref{sec:model}, solved by fixed-point iteration on counterfactual wage changes with the U.S.\ wage as the num\'eraire. Revenue and welfare are in billions of 2025 USD for all exercises.

\subsubsection{Bilateral Laffer Curves}

The bilateral Laffer curve exercise is the basis for Figure~\ref{fig:usa_laffer_eu_chn} (EU and China), Figure~\ref{fig:MEB_no_response}, and the robustness exercises in Supplemental Appendices \ref{app:robustness_elasticities} and \ref{app:additional_laffers}.  

For each partner $j$, we impose a uniform U.S.\ ad-valorem tariff rate $t$ on imports from $j$ in all 10 tariffable sectors, varying $t$ from 0\% to 50\% in 1-percentage-point increments (51 grid points).  All other bilateral tariffs are held at their January~1, 2025 baseline values.  The reference equilibrium sets bilateral free trade between the U.S.\ and $j$ ($\tau_{\text{US},j}^s = \tau_{j,\text{US}}^s = 1$ for all tariffable sectors), with all other tariff wedges unchanged, to avoid differences in bilateral Laffer curves being driven primarily by differences in baseline U.S.\ tariff rates.  Revenue and welfare changes are measured relative to this bilateral free-trade reference, unless explicitly mentioned otherwise.

For each tariff rate $t$, four partner response scenarios are computed:
\begin{enumerate}
\item \textbf{No response}: partner $j$ does not change its tariffs on U.S.\ exports.  
\item \textbf{Equivalent response}: partner $j$ matches the U.S.\ tariff rate $t$ on its imports from the U.S.
\item \textbf{Revenue-maximizing response}: for each U.S.\ rate $t$, partner $j$ sets a uniform tariff on U.S.\ exports that maximizes $j$'s tariff revenue, searched over a grid of $t$ from 0\% to 50\% in 0.1\% increments.
\item \textbf{Welfare-maximizing response}: for each U.S.\ rate $t$, partner $j$ sets a uniform tariff on U.S.\ exports that maximizes $j$'s welfare, searched over a grid of $t$ from 0\% to 50\% in 0.1\% increments.
\end{enumerate}

The dotted vertical lines in each figure mark the uniform tariff rate that maximizes U.S.\ tariff revenue (Laffer peaks) and U.S.\ welfare under each response scenario, on the corresponding plot.  The welfare-maximizing Nash equilibrium under uniform tariffs therefore corresponds to the dotted line on the Welfare plot for the Welfare-maximizing response.

\subsubsection{Cumulative Laffer Curves}

Figure~\ref{fig:cum_Laffer_4panel} varies a single uniform U.S.\ tariff rate $t$ (0--50\%, 1-pp increments) against an expanding set of partners.  Partners are added cumulatively in the order: CHN, CAN, MEX, VNM, TUR, KOR, IND, IDN, EUU, BRA, NOR, CHE, JPN, GBR, AUS, ROW, determined by bilateral MEB at $t = 0$.  For ``$k$ Partners,'' the U.S.\ imposes rate $t$ on imports from the first $k$ partners in this ordering across all tariffable sectors; all other tariffs are held at January~1, 2025 values.  To avoid conflating changes due to U.S.\ tariff rates with changes as the scope of targeted partners change, we report all changes relative to a free trade baseline between the U.S.\ and all partner countries, including for scenarios in which the U.S.\ engages in a trade war for only a subset of countries.

\subsubsection{Marginal Welfare Cost Measures}

The MEB and MFEI curves in Figure~\ref{fig:MEB_no_response} (and Figures~\ref{fig:MEB_opt_response}--\ref{fig:MEB_no_response_alpha50}) are computed from the same exercise as for the bilateral Laffer curve.  Centered finite differences yield approximations for $dW/dt$ and $dR/dt$ at interior grid points.  The marginal excess burden is $\text{MEB}(t) = -(dW/dt)/(dR/dt)$, and the marginal fiscal efficiency index is
\begin{equation*}
\text{MFEI}(t) = \frac{dW/dt + \alpha \cdot dR/dt}{|dW/dt| + \alpha \cdot |dR/dt|},
\end{equation*}
with $\alpha = 0.25$.  Individual partner curves are plotted for all 15 partners (excluding ROW), along with a line representing the median value across partners at each tariff level.  Background shading reflects three zones as defined by the median curve: $dW/dt \geq 0$ and $dR/dt > 0$ (Free-Lunch); $dW/dt < 0$ and $dR/dt > 0$ (Trade-Off); $dW/dt < 0$ and $dR/dt \leq 0$ (Beyond-Laffer).  We suppress the lines connecting the extreme positive and extreme negative values for MEB around its asymptote at the Laffer Peak.

\subsubsection{Timeline of the 2025 U.S.\ Trade War}
\label{app-sec:2025_timeline}
The reference point for Table~\ref{tbl:us_tariff_timeline} is the 2016 MacMap baseline.  The model is calibrated to January~1, 2025 trade flows, with the 2016 equilibrium computed as a counterfactual by imposing 2016 MacMap tariffs.  All changes are reported relative to this 2016 equilibrium using the scaling described in Section~\ref{app-sec:units}.

Each row in the tariff timeline solves two equilibria for a given policy date $d$: one with only U.S.\ tariff changes (partner tariffs at 2016 levels), and one with all countries' tariffs at date $d$.  The ``U.S.\ Only'' and ``$+/-$ w/ Retal.'' columns report these respectively, with the retaliation column showing the incremental change.  These counterfactuals should not be interpreted as estimates of the actual effects on the U.S.\ economy of tariff rates at each date, but rather estimates of what we would have expected the long term annual impact of these tariff rates to be, had they remained fixed at these levels.

Trade-weighted average tariff rates on U.S.\ imports and exports use counterfactual (post-tariff equilibrium) trade shares as weights:
\begin{equation*}
\bar{t}_{\text{imp}} = \frac{\sum_{j \neq \text{US}} \sum_{s \in S_T} (\tau_{\text{US},j}^s - 1) \cdot \pi_{\text{US},j}^s X_{\text{US}}^s}{\sum_{j \neq \text{US}} \sum_{s \in S_T}  \pi_{\text{US},j}^s X_{\text{US}}^s},
\end{equation*}
where $S_T = \{1,\ldots,10\}$ denotes tariffable sectors, evaluated at the equilibrium with retaliation.  The export-side formula is analogous with subscripts reversed.

The bilateral decomposition panel solves, for each named partner $j$, an equilibrium where only tariffs between the U.S.\ and $j$ change from 2016 to the final date, with all other bilateral tariffs at the calibration baseline.  The ``Global Welfare Cost per \$ of U.S.\ Tariff Revenue'' column reports $(|\Delta W_{\text{US}}| + |\Delta W_{\text{RoW}}|) / \Delta R_{\text{US}}$, where $\Delta W_{\text{RoW}}$ includes welfare changes for all non-U.S.\ countries, not only the named partner.  ``Global'' includes the U.S.\ and is used distinctly from ``All'', which includes all U.S.\ partner countries, but excludes the U.S.

\subsection{Inverse-Optimum Specification Details}
\label{app:invopt_details}

\subsubsection{Baselines and Perturbations}

The perturbation exercise (Section~\ref{sec:inverse_optimum}) is run separately for 46 baselines: MacMap tariffs at 2001, 2004, 2007, 2010, 2013, 2016, and 2019 (with 2018 trade war), plus each of the 39 WTO tariff action dates.  At each baseline, and for each bilateral sector-level tariff $\tau_{\text{US},j}^s$, the tariff factor is perturbed up by $\Delta\tau = 0.001$ and the full general-equilibrium response is solved.  This yields one observation per (date $d$, partner $j$, sector $s$) triple, for $15 \times 10 = 150$ observations per baseline (15 partners, excluding ROW, $\times$ 10 tariffable sectors).

The shadow wedge for observation $(j, s, d)$ is
\begin{equation*}
\tilde{\omega}_{j,s,d} = -\frac{\Delta W_{\text{US}}}{\Delta\tau} - \alpha \cdot \frac{\Delta T_{\text{US}}}{\Delta\tau},
\end{equation*}
with $\alpha = 0.25$.  Each observation is classified into one of three MFEI zones: Free-Lunch ($\Delta W_{\text{US}}/\Delta\tau \geq 0$, $\Delta T_{\text{US}}/\Delta\tau > 0$), Trade-Off ($\Delta W_{\text{US}}/\Delta\tau < 0$, $\Delta T_{\text{US}}/\Delta\tau > 0$), or Beyond-Laffer ($\Delta W_{\text{US}}/\Delta\tau < 0$, $\Delta T_{\text{US}}/\Delta\tau \leq 0$).

\subsubsection{Geopolitical Groups}
\label{app:geopolitical_groups}

Partners are assigned to geopolitical groups $g(j)$ based on ideal point estimates for the United Nations 78th General Assembly, which ran from September 2023 through September 2024, from \citet{Baileyetal2017}.\footnote{\url{https://dataverse.harvard.edu/dataset.xhtml?persistentId=doi:10.7910/DVN/LEJUQZ}} We use an ideal point cutoff of $0.65$ (where higher scores indicate closer alignment with the United States) to separate Aligned and Not-Aligned regions, with China and USMCA (CAN, MEX) treated separately.
\begin{itemize}
\item China: CHN.
\item USMCA: CAN, MEX.
\item Aligned: EUU, GBR, JPN, KOR, AUS, CHE, NOR.
\item Not-Aligned: TUR, BRA, IND, IDN, VNM.
\end{itemize}

\subsubsection{Regime Periods}

The 46 baselines are assigned to regime periods $r(d)$ as shown in Table~\ref{tbl:regime_periods}.  Dates not mapped into a regime period are excluded.  Figure~\ref{fig:invopt_timeseries} reports estimates separately for each individual date rather than pooled into these periods.

\begin{table}[htbp]
\centering
\caption{Inverse-Optimum Date Regime Periods}
\label{tbl:regime_periods}
\begin{small}
\begin{tabular}{llr}
\toprule
Period & Baselines & Dates  \\
\midrule
2010-2016 & MacMap 2010, 2013, and 2016 & 3  \\
2019 Trade War & MacMap 2019 with U.S.-China trade war tariffs & 1  \\
Apr 2025--May 2025 & WTO dates Apr~5, 2025 -- May~13, 2025 & 5 \\
Aug 2025--Jan 2026 & WTO dates Aug~17, 2025 -- Jan~1, 2026 & 14 \\
\bottomrule
\end{tabular}
\end{small}
\end{table}

\end{document}